\documentclass[letterpaper, 10 pt, conference]{ieeeconf} 
\let\labelindent\relax
\usepackage[inline]{enumitem}
\usepackage{subcaption}
\IEEEoverridecommandlockouts                              
\usepackage{cite}
\usepackage{amsmath,amssymb,amsfonts}
\usepackage{graphicx}
\usepackage{textcomp}
\usepackage{xcolor}
\usepackage{todonotes}
\def\BibTeX{{\rm B\kern-.05em{\sc i\kern-.025em b}\kern-.08em
		T\kern-.1667em\lower.7ex\hbox{E}\kern-.125emX}}	
\usepackage[T1]{fontenc} 
\usepackage{tikz} 

\usepackage{algorithm} 
\usepackage{graphicx}
\usepackage[utf8]{inputenc}
\usepackage[noend]{algpseudocode}
\usepackage{stackengine} 
\newcommand\barbelow[1]{\stackunder[1.2pt]{$#1$}{\rule{.8ex}{.075ex}}}
\usepackage{multirow} 
\usepackage{pgfplots} 
\usepackage{lipsum} 
\usepackage{subcaption}
\usepackage{nccmath} 
\usepackage{lipsum} 
\usepackage{empheq} 
\usepackage{hyperref}
\usepackage{stackengine} 
\setstackEOL{\\} 
\usepackage{multirow}
\usepackage{adjustbox}
\usepackage{comment}

\usepackage{tcolorbox} 
\colorlet{yellow1}{yellow!15}
\colorlet{alg_box}{gray!15}
\definecolor{darkgray}{rgb}{0.41, 0.41, 0.41}

\newtcolorbox{highlightbox}[1]{
    colback=#1,
    boxrule=0pt,
    boxsep=0pt,
    arc=0pt,
    left=0pt,
    right=0pt,
    top=1pt,
    bottom=1pt,
    width=1.05\linewidth, 
    height=1.3cm,  
}
\usepackage{booktabs} 
\usetikzlibrary{arrows, positioning, shapes.geometric}
\usetikzlibrary{fit, backgrounds} 

\tikzstyle{startstop} = [rectangle, rounded corners, 
minimum width=3cm, 
minimum height=1cm,
text centered, 
draw=black, 
text width=3.6cm, 
fill=red!20]

\tikzstyle{io} = [trapezium, 
trapezium stretches=true, 
trapezium left angle=70, 
trapezium right angle=110, 
minimum width=2.8cm, 
minimum height=1.4cm, 
text centered,
text width=4.5cm, 
draw=black, fill=blue!20]

\tikzstyle{process} = [rectangle, 
minimum width=3cm, 
minimum height=1cm, 
text centered, 
text width=3.8cm, 
draw=black, 
fill=orange!20]

\tikzstyle{decision} = [diamond, 
minimum width=2cm, 
minimum height=0.7cm, 
text centered, 
draw=black, 
text width=2cm, 
fill=green!20]

\tikzstyle{arrow} = [thick,->,>=stealth]
\usepackage{pdfpages}
\usepackage{defs}

\algrenewcommand\algorithmicindent{0.5em} 
\begin{document}
	
	\title{\LARGE \bf A Unifying Complexity-Certification Framework for Branch-and-Bound Algorithms for Mixed-Integer Linear and Quadratic Programming} 

\author{Shamisa Shoja*, Daniel Arnström**, Daniel Axehill* 
	\thanks{*S. Shoja and D. Axehill are
		with the Division of Automatic Control, Department of
		Electrical Engineering, Linköping University, Sweden. Email: 
		{\tt\small \{shamisa.shoja,  daniel.axehill\}@liu.se}.  **D.Arnström is with the Division of Systems and Control, Department of Information Technology, Uppsala University, Sweden. Email: {\tt\small daniel.arnstrom@it.uu.se}.  
  \protect\newline
  This work was partially supported by the Wallenberg AI, Autonomous
Systems and Software Program (WASP) funded by the Knut and Alice
Wallenberg Foundation.}
}

\maketitle
\thispagestyle{empty}
\pagestyle{empty}

		\begin{abstract}
In model predictive control (MPC) for hybrid systems, solving optimization problems efficiently and with guarantees on worst-case computational complexity is critical to satisfy the real-time constraints in these applications. These optimization problems often take the form of mixed-integer linear programs (MILPs) or mixed-integer quadratic programs (MIQPs) that depend on system parameters. A common approach for solving such problems is the branch-and-bound (B\&B) method. This paper extends existing complexity certification methods by presenting a unified complexity-certification framework for B\&B-based MILP and MIQP solvers, specifically for the family of multi-parametric MILP and MIQP problems that arise in, e.g., hybrid MPC applications. The framework provides guarantees on worst-case computational measures, including the maximum number of iterations or relaxations B\&B algorithms require to reach optimality. It systematically accounts for different branching and node selection strategies, as well as heuristics integrated into B\&B, ensuring a comprehensive certification framework.
	By offering theoretical guarantees and practical insights for solver customization, the proposed framework enhances the reliability of B\&B for real-time application. The usefulness of the proposed framework is demonstrated through numerical experiments on both random MILPs and MIQPs, as well as on MIQPs arising from a hybrid MPC problem. 
		\end{abstract}

		\newtheorem{lemma}{Lemma}
		\newtheorem{corollary}{Corollary}
		\newtheorem{theorem}{Theorem}
		\newtheorem{assumption}{Assumption}
		\newtheorem{remark}{Remark}
		\newtheorem{definition}{Definition}	
		\newtheorem{properties}{Property}				

%
\section{Introduction} \label{sec_intro_uni}
Model predictive control (MPC) involves solving an optimization problem at each time step to compute an optimal control action. For safety-critical systems operating in real time, the optimization solvers used in MPC must be both reliable and efficient, especially in embedded systems with limited computational resources and memory. 
When MPC is used to control hybrid systems where certain states and/or controls are constrained to binary values, the resulting optimization problems are typically in the form of mixed-integer linear programs (MILPs) when $1$- or $\infty$-norm performance measures are used in the objective function. If a quadratic objective function is used instead, the optimization problem in hybrid MPC  becomes a mixed-integer quadratic program (MIQP)~\cite{bemporad1999control}. These optimization problems depend on parameters such as the current system state, making them multi-parametric MILPs (mp-MILPs) or mp-MIQPs.
In explicit hybrid MPC, mp-MILPs and mp-MIQPs are typically solved offline for a set of parameters, with the precomputed solutions stored for online use~\cite{dua2000algorithm, bemporad2002explicit}. 
However, as the problem dimensions increase, the complexity of these precomputed solutions grows exponentially, making them impractical for high-dimensional problems due to memory limitations. In such cases, the optimization problem must be solved online instead. This necessitates an efficient solver that can provide
guarantees on solving the problem within the time constraints imposed in real time. Such guarantees on the worst-case computational complexity for solving LPs and QPs encountered in linear MPC have been provided in, e.g.,~\cite{zeilinger2011real}, and~\cite{cimini2017exact, arnstrom2021unifying}, respectively. 

A widely used approach for solving mixed-integer problems is branch and bound (\bnb)~\cite{wolsey2020integer}, where a sequence of relaxations of the problem is solved. 
Recent research~\cite{shoja2022exact, axehill2010improved, shoja2022overall} has focused on certifying the worst-case computational complexity of \bnb methods, particularly for the family of multi-parametric MILP and MIQP problems that arise in, e.g., hybrid MPC applications. These works have established key computational bounds, including the accumulated number of linear systems of equations solved in relaxations (\textit{iterations}) and the total number of relaxations (nodes) explored within B\&B.
This paper extends the results in~\cite{shoja2022overall}, which provided upper bounds on the computational complexity of MIQPs, by deriving exact complexity measures. This development facilitates integrating extensions originally proposed for MILPs in~\cite{shoja2023subopt, shoja2023sheuristic} into the MIQP setting, allowing the certification results of~\cite{shoja2022overall, shoja2022exact, shoja2023subopt, shoja2023sheuristic} to be integrated in a unified framework. 

The main contribution of this paper is, hence, a unified certification framework for standard B\&B algorithms applicable to both MILPs and MIQPs. 
The framework is capable of providing exact worst-case bounds on the total number of iterations and B\&B nodes and can be extended to floating point operations (flops), providing critical insight into solution times. Additionally, it is extended to incorporate various branching strategies and improvement heuristics. As a result, the proposed unified framework accommodates different branching strategies, node-selection strategies~\cite{shoja2022exact}, start heuristics~\cite{shoja2023sheuristic}, and improvement heuristics, making it more comprehensive. 
By providing worst-case guarantees, the framework enhances the reliability of B\&B solvers in real-time applications. Furthermore, explicit knowledge of encountered subproblems enables solver customization, thereby reducing online computational costs.

When computing exact complexity measures for B\&B-based MIQP solvers, additional challenges arise due to the non-polyhedral nature of parameter-space partitioning, which stems from quadratic function comparisons. To address this, we extend the results in~\cite{shoja2022overall} by introducing conservative approximation techniques for these comparisons, thereby mitigating the impact of non-polyhedral geometry.

In summary, the main contributions of this paper are:
\begin{itemize}[leftmargin=*]
\item 
\textit{A unified complexity-certification framework} for B\&B algorithms, applicable to both MILPs and MIQPs.
\item \textit{Integration of algorithmic strategies} into the framework, including branching strategies, node-selection strategies, and primal heuristics, enhancing its applicability.
\item  \textit{Development of quadratic function approximations} to address non-polyhedral partitioning for MIQPs, yielding a conservative yet computationally tractable certification framework.
\end{itemize}

\section{Problem formulation}  \label{sec:prob_form_uni}
When applying MPC to hybrid systems with  a $1$-norm or $\infty$-norm performance measure, the optimization problem can be formulated as an mp-MILP (see, e.g.,~\cite{borrelli2017predictive}) of the form 
\begin{subequations} \label{eq:mpMILP_uni}
	\begin{align}
		\min_{x} \quad &  c^T x \label{eq:mpMILP_uni1} \\ 
		\mathcal{P}_{\text{mpMILP}}(\theta): \hspace{.2cm} \textrm{s.t.} \quad & Ax \leqslant b + W \theta, 
        \label{eq:mpMILP_uni2} \\
		& x_i \in  \{0,1\}, \hspace{.2cm} \forall i \in \mathcal{B}, \label{eq:mpMILP_uni3}
	\end{align}
\end{subequations} %
where  $x \in \mathbb{R}^{n_c} \times \{0, 1\}^{n_b}$ is the vector of in total $n = n_c + n_b$ decision variables comprising $n_c$ continuous and $n_b$ binary variables, 
and $\theta \in \Theta_0 \subset \mathbb{R}^{n_{\theta}}$ is a vector of parameters in the polyhedral parameter set $\Theta_0$.  In hybrid MPC applications, $x$ typically represents the control action, while $\theta$ relates to the system state and reference signal.
The objective function in~\eqref{eq:mpMILP_uni} is defined by $c \in \mathbb{R}^n$, 
and the feasible set is specified by $A \in \mathbb{R}^{m \times n}$, $b \in \mathbb{R}^{m}$, and $W \in \mathbb{R}^{m \times n_{\theta}}$.
Because the variables indexed by the set $\mathcal{B}$ are binary-valued, the problem~\eqref{eq:mpMILP_uni}  is non-convex and classified as $\mathcal{NP}$-hard~\cite{wolsey2020integer}. 

For hybrid MPC with a $2$-norm performance measure, the optimization problem can be cast into an mp-MIQP (see, e.g.,~\cite{bemporad1999control, borrelli2017predictive}) of the form

\begin{subequations} \label{eq:mpMIQP_uni}		
	\begin{align}
		\vspace{.1cm} 
		\min_{x} \quad & \frac{1}{2} x^{T}Hx + f^T x + \theta ^{T} f_{\theta }^{T} x 
        \label{eq:mpMIQP_uni1} \\		
		\mathcal{P}_{\text{mpMIQP}}(\theta): \hspace{.3cm} \textrm{s.t.} \quad & Ax \leqslant b + W \theta, 
        \label{eq:mpMIQP_uni2} \\
		& x_i \in  \{0,1\}, \hspace{.3cm} \forall i \in \mathcal{B}, \label{eq:mpMIQP_uni3}
	\end{align}
\end{subequations}
where $H \in \mathbb{S}_{+}^n$, $f \in \mathbb{R}^n$, and $f_{\theta} \in \mathbb{R}^{n \times n_{\theta}}$. 
The decision variables, the parameter vector, and the feasible set are defined similarly to those in~\eqref{eq:mpMILP_uni}. For brevity, we use the compact notation $b(\theta) = b + W \theta$ and $f(\theta) = f + f_{\theta} \theta$.

Throughout this paper, $\mathbb{N}_{0}$ denotes the set of nonnegative integers, $\mathbb{N}_{1:N}$ represents the finite set $\{1, \dots, N\}$, and $\{\mathcal{C}^i\}_{i=1}^{N}$ denotes a finite collection $\{\mathcal{C}^1, \dots, \mathcal{C}^N\}$ of $N$ elements. When $N$ is unimportant, we use $\{\mathcal{C}^i\}_{i}$ instead. 

\section{Branch and bound} 
This section reviews (online) B\&B methods, providing the detailed background necessary for analyzing the properties of the certification algorithm presented in subsequent sections. A standard reference on this topic is~\cite{wolsey2020integer}. Readers already familiar with the fundamentals of B\&B methods and primal heuristics for B\&B may skip this section and proceed directly to Section~\ref{sec:cert_uni}.

\subsection{Introduction}
Consider the mp-MILP problem~\eqref{eq:mpMILP_uni} and the mp-MIQP problem~\eqref{eq:mpMIQP_uni}, where the parameter vector $\theta$ is fixed at a specific value $\bar{\theta}$. The resulting (non-parametric) instance then becomes an MILP in the form 

\begin{subequations} \label{eq:MILP_uni}			\begin{align} 		\vspace{.1cm} 		\min_{x} \quad & c^T x \label{eq:MILP_uni1} \\ 		\mathcal{P}_{\text{MILP}}(\bar{\theta}): \hspace{.3cm} \textrm{s.t.} \quad & Ax \leqslant  b( \bar{\theta}),
\label{eq:MILP_uni2} \\		& x_i \in  \{0,1\}, \hspace{.3cm} \forall i \in \mathcal{B}, \label{eq:MILP_uni3}	\end{align}\end{subequations}%
or an MIQP in the form 

\begin{subequations} \label{eq:MIQP_uni}		
	\begin{align} 
		\vspace{.1cm} 
		\min_{x} \quad & \frac{1}{2} x^{T}Hx + f(\bar{\theta})^T x 
        \label{eq:MIQP_uni1} \\		
		\mathcal{P}_{\text{MIQP}}(\bar{\theta}): \hspace{.3cm} \textrm{s.t.} \quad & Ax \leqslant b( \bar{\theta}),
        \label{eq:MIQP_uni2} \\
		& x_i \in  \{0,1\}, \hspace{.3cm} \forall i \in \mathcal{B}, \label{eq:MIQP_uni3}
	\end{align}
\end{subequations}%
respectively. 
A naive approach to solve MILPs and MIQPs is to explicitly enumerate all \(2^{|\mathcal{B}|}\) combinations of binary variables in \(\mathcal{B}\) and solve the resulting LP or QP for each combination using methods such as the simplex or active-set methods~\cite{nocedal2006numerical}, where $|.|$  denotes the cardinality of the set. While this guarantees that an optimal solution (if one exists) will be found, it quickly becomes computationally intractable as \(|\mathcal{B}|\) grows. In contrast, the B\&B method~\cite{land1960automatic} offers a more efficient alternative by implicitly exploring the solution space through a search tree, progressively fixing constraints in \(\mathcal{B}\) to navigate the relaxations efficiently.

In B\&B, solving the MILP problem~\eqref{eq:MILP_uni} 
involves solving a sequence of LP relaxations in the form
\begin{subequations} \label{eq:LP_uni}
	\begin{align}
		\vspace{.1cm} 
		\min_{x} \quad & c^T x 
		\label{eq:LP_uni1} \\ \mathcal{P}_{\text{LP}}(\bar{\theta}): \hspace{.3cm} \textrm{s.t.} \quad & Ax \leqslant  b(\bar{\theta}), 
        \label{eq:LP_uni2} \\
		& 0 \leqslant x_i  \leqslant 1, \hspace{.2cm} \forall i \in \mathcal{B}, \label{eq:LP_uni3}\\
		& x_i = 0, \hspace{.08cm} \forall i \in \mathcal{B}_0, \hspace{.1cm} x_i = 1, \hspace{.08cm} \forall i \in \mathcal{B}_1, \label{eq:LP_uni4}
	\end{align}
\end{subequations} %
where $\mathcal{B}_0 \cap \mathcal{B}_1 = \emptyset$ and $\mathcal{B}_0,\mathcal{B}_1 \subseteq \mathcal{B}$ represent index sets of binary variables fixed to $0$ and $1$, respectively. Similarly, when solving the MIQP problem~\eqref{eq:MIQP_uni} in B\&B, QP relaxations $\mathcal{P}_{\text{QP}}(\bar{\theta})$ are solved in \bnb, where they are derived analogously by relaxing the binary constraints~\eqref{eq:MIQP_uni3} into~\eqref{eq:LP_uni3}.
These relaxations are systematically ordered in a B\&B search tree, with each relaxation forming a \textit{node}. The following definition formalizes the concept of a node in the B\&B search tree.

\begin{definition}[Node]
	Let a \textit{node} be characterized by a tuple $\eta \triangleq \left( \mathcal{B}_0, \mathcal{B}_1\right)$, where $\mathcal{B}_0$ and $\mathcal{B}_1$ are defined in \eqref{eq:LP_uni4}. Then, we can define the following:
	\begin{itemize}
		\item The \textit{level} of a node is defined by $l (\eta) \triangleq |\mathcal{B}_0| + |\mathcal{B}_1|$.
		\item  Node $\eta = (\emptyset, \emptyset)$ at the top of the tree 
		is called the \textit{root} node. 
		Moreover, nodes at the bottom of the tree for which $l(\eta) = n_b$ are called \textit{leaf} nodes. 
		\item Node $\hat{\eta} = ( \hat{\mathcal{B}_0}, \hat{\mathcal{B}_1} )$ is a \textit{descendant} to node $\eta = \left( \mathcal{B}_0, \mathcal{B}_1\right)$, denoted by $\hat{\eta} \in \mathcal{D}(\eta)$, if $\hat{\mathcal{B}_0} \supseteq \mathcal{B}_0$ and $\hat{\mathcal{B}_1} \supseteq \mathcal{B}_1$. Furthermore, the node $\hat{\eta} \in \mathcal{D}(\eta)$ is a \textit{child} to (the \textit{parent} node) $\eta$ if $l(\hat{\eta}) - l (\eta) = 1$.
	\end{itemize}
\end{definition}

\subsection{Algorithm overview}
The generic online \bnb algorithm (for a given $\theta$) is outlined in Algorithm~\ref{alg:BnB_on_uni}. Here,  $\mathcal{T}$ represents a sorted list of \textit{pending nodes} that are yet to be explored. Exploring (or processing) a node involves solving the corresponding (LP/QP) relaxation. In the algorithm, $\barbelow{J}$ denotes the objective function value of the relaxation (lower bound), $\barbelow{x}$ denotes the corresponding optimal solution, and $\mathcal{A}$ denotes the active set (i.e., the set of constraints satisfied with equality at $\barbelow{x}$). 
Additionally, $\kappa$ denotes the complexity measure, which can be, e.g., the number of linear systems of equations solved in the relaxation (the \textit{iteration} number) or the number of relaxations (B\&B nodes) explored. 
Moreover, $\bar{J}$ and $\bar{x}$ denote the objective function value (upper bound) and the solution of the best-known integer-feasible solution found so far. 
As a reference, Algorithm~\ref{alg:BnB_on_uni} also stores and returns the accumulated complexity measure, denoted $\kappa^*_\text{tot}$, which will later be compared with the complexity measure computed by the certification framework. 

\begin{algorithm}[htbp]
	\caption{\textsc{\bnb}: Generic (online) \bnb algorithm} 
	\label{alg:BnB_on_uni}
	\begin{algorithmic}[1]
		\Require MILP/MIQP problem~\eqref{eq:MILP_uni}/\eqref{eq:MIQP_uni} 
        (for a given $\theta$) 
		\vspace{.005cm}
		\Ensure $\bar{J}$, $\bar{x}$, $\kappa^*_\text{tot}$ 
        \vspace{.06cm}
		\State{
			$\mathcal{T} \leftarrow \{(\emptyset, \emptyset)\}$, $\kappa^*_\text{tot} \leftarrow 0$, $\bar{J} \leftarrow \infty$, $\bar{x} \leftarrow \texttt{NaN} $}
		\While {$\mathcal{T} \neq \emptyset$} \label{step:loop_T_on_uni}
		\State{$ \eta  \leftarrow$ pop node from $\mathcal{T}$} \label{step:pop_on_uni} 
		\State ($\kappa, \mathcal{A}$,  $\barbelow{J}, \barbelow{x}$)  $\leftarrow$ $\textsc{solve}(\eta) $ 
		\label{step:solve_on_uni} %
		\State $\kappa^*_\text{tot} \leftarrow  \kappa^*_\text{tot} + \kappa$ \label{step:kappa_on_uni}
		\vspace{.05cm}
		\State $ \mathcal{T}, \bar{J}, \bar{x} \leftarrow \textsc{cut}$$ \left(\mathcal{T},  \bar{J}, \bar{x}, \barbelow{J}, \barbelow{x}, \mathcal{A}, \eta \right)$ 
	\label{step:cutcond_on_uni} 	
		\EndWhile 
		\State \textbf{return} 
        $\bar{J}$, $\bar{x}$, $\kappa_\text{tot}^*$  \label{Step:BnB_on_return_uni}
	\end{algorithmic}
\end{algorithm}

Algorithm \ref{alg:BnB_on_uni} performs the following steps at each iteration: 
\begin{enumerate}
	\item \textit{Selecting the first pending node from the list} \hspace{.05cm} (Step~\ref{step:pop_on_uni}) if it exists; otherwise, terminate. 
	\item \textit{Solving the relaxation} \hspace{.05cm} (Step~\ref{step:solve_on_uni}) using an LP/QP solver, such as the simplex/active-set methods~\cite{nocedal2006numerical}.
	\item \textit{Bounding or branching the node} \hspace{.05cm} (Step~\ref{step:cutcond_on_uni}) to either prune the node or generate two child nodes, using the \textsc{cut} procedure given in Algorithm~\ref{alg:cut_cond_on_uni}. Details of this procedure are provided in Section~\ref{subsec:cut_cond_on_uni}. 
\end{enumerate}

\subsection{Cut-condition evaluation} \label{subsec:cut_cond_on_uni}
In B\&B, solutions to previous relaxations can sometimes be used to dismiss further expansion of a subtree, known as \textit{cuts}. These cuts rely on the following lemma:
\begin{lemma} \label{lem:decentent_uni} 
	Let $\barbelow{J}_{\eta}$ denote the objective function value of a relaxation at node $\eta$, and let $\hat{\eta} \in \mathcal{D}(\eta)$. Then $\barbelow{J}_{\hat{\eta}} \geq \barbelow{J}_{\eta}$.
\end{lemma}
\begin{proof}
	It follows directly from standard \bnb arguments, as the feasible set for the relaxation at a descendant node $\hat{\eta}$ is a subset of the feasible set at node $\eta$.
\end{proof} 

The \textsc{cut} procedure 
is outlined in Algorithm~\ref{alg:cut_cond_on_uni}.
The two types of cut conditions applied in this algorithm are as follows:
\begin{itemize}[leftmargin=*]
	\item \textit{Dominance} \hspace{.02cm} (Step \ref{step:dominc_on_uni}): 
	If the objective function value of a relaxation ($\barbelow{J}$) is greater than that of the best-known integer solution ($\bar{J}$), further exploration of descendant nodes is unnecessary, as from Lemma~\ref{lem:decentent_uni}, they cannot yield a better solution. \textit{Infeasibility} can be considered a special case of this cut, where $\barbelow{J}$ is set to infinity for infeasible problems.
	
	\item \textit{Integer feasibility} \hspace{.02cm} (Step~\ref{step:int_cut_on_uni}): If the solution to a relaxation satisfies the binary constraints, no further descendants need to be explored, as from Lemma~\ref{lem:decentent_uni}, any additional branching cannot yield better integer-feasible solutions. 
\end{itemize}

\begin{algorithm}[htbp]	\caption{\textsc{cut}: Cut-condition evaluation and branching in the online B\&B} 
	\label{alg:cut_cond_on_uni}	
	\begin{algorithmic}[1] 
		\Require \Longunderstack[l]{$\mathcal{T}$, $\bar{J}$, $\bar{x}$, $\barbelow{J}$, $\barbelow{x}$, $\mathcal{A}$,  $(\mathcal{B}_0, \mathcal{B}_1)$} 
		\vspace{.005cm}
		\Ensure \Longunderstack[l]{ $\mathcal{T}$, $\bar{J}$, $\bar{x}$}
		\vspace{.06cm}
		\If {$\barbelow{J} \geq \bar{J} $} \label{step:dominc_on_uni}
		\State No feasible solution better than $\bar{x} $
		\label{step:dominc_on2_uni}
        \Comment{\textcolor{darkgray}{\textit{dominance cut}}}
		\ElsIf {all relaxed binary constraints are active} 
		 \label{step:int_cut_on_uni}
		\State $\bar{J} \leftarrow \barbelow{J}$, \hspace{.05cm} $\bar{x} \leftarrow \barbelow{x} $ \Comment{\textcolor{darkgray}{\textit{integer-feasibility cut}}}
		\label{step:intfeas_on_uni}
		\Else
		\Comment{\textcolor{darkgray}{\textit{branching}}}\label{step:cut_on_else} 
		\State{$k \leftarrow$ \textsc{branchInd}($\barbelow{x}, \mathcal{B}$)} 
		\label{step:branch_on_uni}	
		
	\State \Longunderstack[l]{$\mathcal{T} \leftarrow$ \textsc{sort}$\left((\mathcal{B}_0 \cup \{k\},\mathcal{B}_1 ), (\mathcal{B}_0,\mathcal{B}_1 \cup \{k\}), \barbelow{J}, \mathcal{T} \right)$} 
		\label{step:append_on_uni} 
		\EndIf	
		\State \textbf{return} $\mathcal{T}$, $\bar{J}$, $\bar{x}$  \label{step:cut_on_return_uni}
	\end{algorithmic}
\end{algorithm}

If no cut conditions apply, an index $k$  of a relaxed binary constraint is selected (Step~\ref{step:branch_on_uni}), and the two child nodes created by fixing variable $k$ are then stored in  $\mathcal{T}$ (Step~\ref{step:append_on_uni}).  The procedures for selecting branching indices and sorting nodes are detailed in Sections~\ref{subsec:branching_on_uni} and~\ref{subsec:node_selec_on_uni}, respectively. 

\subsection{Branching strategies} \label{subsec:branching_on_uni}
An important choice in \bnb is selecting which variable to branch on. 
The objective of an effective branching strategy is to minimize the number of nodes explored in the search tree. 
Let $\bar{\mathcal{B}} \triangleq \{i \in \mathcal{B} \hspace{.1cm} | \hspace{.1cm} \barbelow{x}_i \notin \{0,1\} \}$ denote the set of \textit{branching candidates}. 
 The \textsc{branchInd} procedure, outlined in Algorithm~\ref{alg:branching_on_uni}, selects the branching index by assigning a \textit{score value} 
$s_i \in \mathbb{R}$ to each candidate $i \in \bar{\mathcal{B}}$ (Step~\ref{step:branch_score_on_uni}) based on the chosen branching strategy.  The candidate with the highest score is then selected (Step~\ref{step:select_branch_on_uni}). 

\begin{algorithm}[htbp]
	\caption{\textsc{branchInd}: Branching-index selection in the online \bnb} \label{alg:branching_on_uni}
	\begin{algorithmic}[1]
		\Require $\barbelow{x}$, $\mathcal{B}$ 
		\Ensure $k$ 
		\vspace{.06cm}
		\State $\bar{\mathcal{B}} \leftarrow $ $\{i \in \mathcal{B} \hspace{.1cm} | \hspace{.1cm} \barbelow{x}_i \notin \{0,1\} \}$
		\State{\textbf{for} $i \in \bar{\mathcal{B}}$, compute $s_i$} \label{step:branch_score_on_uni}	
		\State Choose $k$ such that $s_k \geq s_i$, $ \forall i \in \bar{\mathcal{B}} \setminus \{k\} $ 
        \label{step:select_branch_on_uni}
		\State \textbf{return} $k$	
	\end{algorithmic}
\end{algorithm}

Various branching strategies have been studied  in~\cite{achterberg2005branching, morrison2016branch}. One such strategy is the \emph{most infeasible} branching (MIB) strategy, which selects the relaxed binary variable closest to 
$0.5$ (i.e., the most infeasible one). For this strategy, the score for each candidate $i \in \bar{\mathcal{B}}$ is computed as:
\begin{equation} \label{eq:branch_minf_on_uni}
	s_i = 0.5 - | \barbelow{x}_{i} - 0.5 |. 
\end{equation}

\subsection{Node-selection strategies} \label{subsec:node_selec_on_uni} 
Another key decision in \bnb is selecting the next pending node for processing. A \textit{sorting criterion}, denoted $\rho(\cdot)$, can determine the processing order of nodes based on the chosen node-selection strategy, where lower values of $\rho(\cdot)$ indicate higher processing priority. Common node-selection strategies include~\cite{wolsey2020integer}:
\begin{itemize}[leftmargin=*]
	\item \textit{Depth-first} (DF): Selects nodes based on depth level, using $\rho_{\scriptscriptstyle \text{DF}}(\eta) =  \frac{1}{l(\eta) + 1}$, favoring nodes at greater depths that are more likely to yield integer-feasible solutions.
	\item \textit{Breadth-first} (BrF): Processes nodes level by level, following $\rho_{\scriptscriptstyle \text{BrF}}(\eta) = l(\eta) + 1$, ensuring that all nodes at a given depth are explored before moving to deeper levels.
	\item \textit{Best-first} (BF): Prioritizes nodes with the lowest relaxation objective function values, using $\rho_{\scriptscriptstyle \text{BF}}(\eta) = \barbelow{J}$, encouraging exploration of nodes practically 
    promising to yield optimal solutions.
\end{itemize}

The \textsc{sort} procedure, outlined in  Algorithm~\ref{alg:append_on_uni}, sorts and stores new nodes $\tilde{\eta}_{0}$ and $\tilde{\eta}_{1}$ 
into $\mathcal{T}$ based on the chosen node-selection strategy. 
Here, $\eta_i$ denotes  the $i$th node in $\mathcal{T}$. 
The algorithm compares the sorting criterion $\rho(\tilde{\eta}_0)$ ($=\rho( \tilde{\eta}_1)$) of the new nodes with those of existing nodes in $\mathcal{T}$, sorting them in ascending 
$\rho$ values to ensure nodes with the highest priority (lowest $\rho$) are processed first. 
This priority structure assumes that the first node is popped from $\mathcal{T}$ at Step~\ref{step:pop_on_uni} of Algorithm~\ref{alg:BnB_on_uni}. 
Whether the \( 0 \)-branch ($\tilde{\eta}_0$) or the \( 1 \)-branch ($\tilde{\eta}_1$) is explored first depends on the order in which they are inserted into  $\mathcal{T}$ (see Step~\ref{step:append_on_pushT_uni}). 

\begin{algorithm}[htbp]
	\caption{\textsc{sort}: 
		Sorted node insertion in the online \bnb} 
	\label{alg:append_on_uni}
	\begin{algorithmic}[1]
		\Require $\tilde{\eta}_{0}, \tilde{\eta}_{1}$, $\barbelow{J}$, $\mathcal{T} = \{\eta_i\}_{i=1}^{N_{\mathcal{T}}} $  
		\Ensure $\mathcal{T}$ 
		\vspace{.06cm}
		\State Compute  $\rho(\tilde{\eta}_0)$ based on the node-selection strategy 
		\label{step:append_on_updateO_uni}
		\For {$i \in \mathbb{N}_{1:N_{\mathcal{T}}}$} 
	\If {$\rho(\tilde{\eta}_0)  \leq  \rho(\eta_i) $} \label{step:append_on_comp_uni} 
    \State{$\mathcal{T} = \{ \{\eta_j\}_{j=1}^{i-1}, \tilde{\eta}_{0}, \tilde{\eta}_{1},\{\eta_j\}_{j=i}^{N_{\mathcal{T}}}\}$}\label{step:append_on_pushT_uni} 
	\State \textbf{return} $\mathcal{T}$ 		
	\EndIf
	\EndFor
\end{algorithmic}
\end{algorithm}

\subsection{Heuristics in
B\&B} \label{sec:heuris_on_uni}
To enhance the performance of \bnb, primal heuristic methods are often employed in high-performance solvers. These methods are \textit{incomplete}, meaning they do not guarantee finding a feasible solution. Nonetheless, they have shown to be highly effective, especially in finding feasible solutions early in the \bnb process. 
Heuristic methods for \bnb have been extensively studied, in, e.g.,~\cite{glover1997general, berthold2006primal, achterberg2007constraint}. Below, two key categories of primal heuristics are reviewed. 

\subsubsection{Start heuristics} \label{subsec:st_heu_on_uni}
Start heuristics aim to identify an integer-feasible solution after solving the initial relaxation (root node). 
Finding feasible solutions early can result in a useful upper
bound that can help to prune some nodes in the B\&B search
tree, thus potentially reducing the tree size and the overall effort. Examples of such heuristics include relaxation enforced
neighborhood search (RENS), diving, and feasibility pump
(FP)~\cite{berthold2006primal}. 

\subsubsection{Improvement heuristics} \label{subsec:imp_heu_on_heu}
Improvement heuristics leverage existing integer-feasible solutions $\bar{x}$ to generate improved solutions $\hat{x}$ with lower objective function values. Examples of such heuristics include \textit{local branching} (LB)~\cite{fischetti2003local} and \textit{relaxation-induced neighborhood search} (RINS)~\cite{danna2005exploring}. These methods construct and solve mixed-integer subproblems using available information to identify improved solutions $\hat{x}$. To guarantee improvement over $\bar{x}$, the subproblems incorporate an additional constraint derived from the objective function (the objective cut-off constraint).
When applying LB and RINS  to MIQPs, incorporating the quadratic objective into a constraint leads to mixed-integer quadratically constrained quadratic program (MIQCQP) subproblems. To maintain the tractability of these heuristics—particularly within the certification framework—we restrict the use of LB and RINS to MILPs in this work,  reviewed below.

\vspace{.1cm} 
\textit{Local branching} \label{subsubsec:LB_heuris_on_uni} \hspace{.2cm}
LB is a refinement heuristic that exploits the proximity of feasible solutions in their Manhattan-distance neighborhood~\cite{fischetti2003local}. Starting from \( \bar{x} \), LB seeks a better solution \( \hat{x} \) within the \(r_n\)-neighborhood of \( \bar{x} \), where \( r_n \in \mathbb{N} \), by solving sub-MILP problems \( \hat{\mathcal{P}}_{\text{LB}} \) derived by adding the following constraints to~\eqref{eq:MILP_uni}~\cite{berthold2006primal}:
\begin{itemize} 
\item \textit{Local branching cut}, to enforce the $r_n$-neighborhood constraint: 
\begin{equation} \label{eq:LB_neigh_on_uni}
	\sum_{i \in \mathcal{B} : \bar{x}_i = 0} x_i + \sum_{i \in \mathcal{B} : \bar{x}_i = 1} (1 - x_i) \leq r_n. 
\end{equation}
\item \textit{Objective cut-off constraint}, to ensure an improved objective value:  
\begin{equation} \label{eq:LB_cutoff_on_uni}  
	c^T x \leq (1 - \varepsilon) c^T \bar{x}, \quad \varepsilon > 0.  
\end{equation}  
\end{itemize}  

The \textsc{LB} heuristic for online B\&B is outlined in Algorithm~\ref{alg:heuris_LB_on_uni}. The sub-MILPs are solved at Step~\ref{step:LB_on_solve} using the online \textsc{B\&B}  (Algorithm~\ref{alg:BnB_on_uni}). The computational effort to solve these subproblems can be limited by imposing constraints on, e.g., the maximum number of \bnb nodes solved. The neighborhood size $r_n$ is dynamically adjusted in the \textsc{neighborSize} procedure based on the subproblem results (Step~\ref{step:LB_kneigh_on}). 
The accumulated complexity measure \( \hat{\kappa}^{h*} \) for the heuristic is tracked and returned by  Algorithm~\ref{alg:heuris_LB_on_uni}. 

\begin{algorithm}[htbp]
\caption{\textsc{LB}: Local branching heuristic for the online \bnb} 
\label{alg:heuris_LB_on_uni}
\begin{algorithmic}[1]
	\Require 
	MILP problem~\eqref{eq:MILP_uni} 
    (for a given $\theta$), $r_{n_0}$, $\bar{x}$ 
	\Ensure $\hat{J}$, $\hat{x}$, $\hat{\kappa}^{h*}$ 
	\vspace{.06cm}
	\State $\hat{J}\gets\infty,\hat{x}\gets\texttt{NaN},\hat{\kappa}^{h*}\gets0$, $r_{n}\gets r_{n_0},\texttt{exec}\gets\texttt{true}$ 
	\While{$\texttt{exec}$}
	\label{step:LB_Nmax_on}
	\State \Longunderstack[l]{Formulate subproblem $\hat{\mathcal{P}}_{\text{LB}}$ by adding constraints~\eqref{eq:LB_neigh_on_uni}\\ and~\eqref{eq:LB_cutoff_on_uni} to the original problem~\eqref{eq:MILP_uni}} 
	\State $\hat{J}, \hat{x}, \hat{\kappa} \gets \textsc{B\&B}(\hat{\mathcal{P}}_{\text{LB}})$
	\label{step:LB_on_solve}
	\State $\hat{\kappa}^{h*} \gets \hat{\kappa}^{h*} + \hat{\kappa}$
	\label{step:LB_on_kappa}
	\If{$\hat{x} \neq \bar{x}$} 
	\Comment{\textcolor{darkgray}{\textit{new solution}}}
	\State \textbf{break}
	\label{step:LB_on_sol}
	\EndIf
	\State $\texttt{exec}, r_n \gets \textsc{neighborSize}(\texttt{exec}, r_n, r_{n_0}, \hat{J})$
	\label{step:LB_kneigh_on}
	\EndWhile
	\State \Return $\hat{J}, \hat{x}, \hat{\kappa}^{h*}$
	\\ \hrulefill
	\Procedure{kNeighbor}{$\texttt{exec}, r_n, r_{n_0}, \hat{J}$}
	\vspace{.06cm}
	\If{$\hat{J} = \infty$} 
	\Comment{\textcolor{darkgray}{\textit{subproblem not solved within limits}}} 
	\If{$r_n=r_{n_0}$}
	\State $r_{n} \gets r_{n_0} - \left\lfloor \frac{r_{n_0}}{2} \right\rfloor$ 
	\Comment{\textcolor{darkgray}{\textit{reduce neighborhood size}}}
	\Else
	\State $\texttt{exec} \gets \texttt{false}$
	\EndIf
	\Else
	\If{$r_{n} = r_{n_0}$}
	\State $r_{n} \gets r_{n_0} + \left\lceil \frac{r_{n_0}}{2} \right\rceil$
	\Comment{\textcolor{darkgray}{\textit{increase neighborhood size}}}
	\Else
	\State $\texttt{exec} \gets \texttt{false}$
	\EndIf
	\EndIf
	\State \Return $\texttt{exec}, r_{n}$
	\EndProcedure
\end{algorithmic}
\end{algorithm}

\vspace{.1cm} 
\textit{RINS} \label{subsec:rins_on_heu} \hspace{.2cm}
RINS exploits the similarity between the integer-feasible solution $ \bar{x} $ and the relaxation solution $\barbelow{x}$~\cite{danna2005exploring}. 
Given \( \bar{x} \) and \( \barbelow{x} \), RINS seeks a better solution \( \hat{x} \) by solving a sub-MILP \( \hat{\mathcal{P}}_{\text{RINS}} \) derived by adding the objective cut-off constraint~\eqref{eq:LB_cutoff_on_uni} and the following constraint to~\eqref{eq:MILP_uni}:
\begin{equation}
	x_i = \bar{x}_i, \hspace{.3cm} \forall i \in \mathcal{B} \hspace{.1cm} \text{ such that } \hspace{.1cm} \bar{x}_i = \barbelow{x}_i, \label{eq:Rins_add_cons_heu}
\end{equation}
to fix all variables in \( \mathcal{B} \) with identical values in \( \bar{x} \) and \( \barbelow{x} \).

The \textsc{RINS} heuristic for online B\&B is outlined in Algorithm~\ref{alg:rins_on_uni}. As with LB, limits (e.g., on the maximum number of nodes) can be imposed on  \textsc{B\&B}  (Algorithm~\ref{alg:BnB_on_uni}) at Step~\ref{step:rins_on_solv_heu}. Additionally, this heuristic is applied only if a sufficient fraction \( r \in [0,1] \) of relaxed binary variables can be fixed (see Step~\ref{step:rins_on_if_heu}),  thereby reducing the risk of solving a sub-MILP that is more difficult than the original MILP~\cite{berthold2006primal}. 
		\begin{algorithm}[htbp]
			\caption{\textsc{RINS}: RINS heuristic for the online \bnb} 
		\label{alg:rins_on_uni}
		\begin{algorithmic}[1]
			\Require MILP problem~\eqref{eq:MILP_uni} 
    (for a given $\theta$), $\barbelow{x}$, $\bar{x}$ 
	\Ensure $\hat{J}$, $\hat{x}$, $\hat{\kappa}^{h*}$ 
	\vspace{.06cm} 
		\State $\hat{J} \gets \infty$, $\hat{x} \gets \texttt{NaN}$, $\hat{\kappa}^{h*} \gets 0$	\State{$\hat{\mathcal{B}} \leftarrow \{i \in \mathcal{B} \hspace{.1cm} | \hspace{.1cm} \bar{x}_i = \barbelow{x}_i \hspace{.005cm}\}$} \label{step:rins_on_B_heu}
        \If{$|\hat{\mathcal{B}}| \geq r |\mathcal{B}|$} \label{step:rins_on_if_heu} 
			
            \State \Longunderstack[l]{Formulate subproblem $\hat{\mathcal{P}}_{\text{RINS}}$ by adding constraints~\eqref{eq:LB_neigh_on_uni}\\ and~\eqref{eq:LB_cutoff_on_uni} to the original problem~\eqref{eq:MILP_uni}} 
			\vspace{.05cm}
			\State $\hat{J}, \hat{x}, \hat{\kappa}^{h*}\gets \textsc{B\&B}(\hat{\mathcal{P}}_{\text{RINS}})$ 
            \label{step:rins_on_solv_heu}
			\EndIf
			\State \textbf{return} $\hat{J}, \hat{x},  \hat{\kappa}^{h*} $  \label{Step:rins_on_return_heu}
		\end{algorithmic}
	\end{algorithm}

\subsubsection{Integration of heuristics into the online B\&B algorithm} \label{subsubsec:LB_heu_integ_on_uni}
A start heuristic can be invoked in the online \textsc{B\&B} Algorithm~\ref{alg:BnB_on_uni} after solving the root node (i.e., between Steps~\ref{step:solve_on_uni} and~\ref{step:kappa_on_uni}).  
Similarly, an improvement heuristic can be applied after finding an integer-feasible solution $ \bar{x} $ (between Steps~\ref{step:intfeas_on_uni} and~\ref{step:cut_on_else} in Algorithm~\ref{alg:cut_cond_on_uni}). These heuristics may be invoked once or multiple times during the \bnb process.  
If a heuristic identifies a better solution $ \hat{x} \neq \bar{x} $, the current best solution $ \bar{x} $ and upper bound $ \bar{J} $ are updated to $ \hat{x} $ and $\hat{J}$, respectively. 
Additionally, the complexity measure $\hat{\kappa}^{h*}$ associated with the heuristic is added to the accumulated complexity measure \( \kappa_{\text{tot}}^* \) of \textsc{B\&B} to account for the overall computational effort.

\section{Unified complexity-certification framework}
\label{sec:cert_uni}

In this section, we present a unified framework to determine the complexity measures induced by different parameters  $\theta$ when applying Algorithm~\ref{alg:BnB_on_uni} to the mp-MILP problem~\eqref{eq:mpMILP_uni} or the mp-MIQP problem~\eqref{eq:mpMIQP_uni}. This approach enables analyzing the algorithm's complexity by computing a complexity measure $\kappa(\theta)$ as a function of $\theta$. Extending the work in~\cite{shoja2022overall, shoja2022exact} and drawing on partitioning techniques similar to~\cite{arnstrom2021unifying}, this method systematically partitions the parameter space based on changes in the solver’s state, allowing for an exact characterization of the \bnb solver's behavior. In particular, some operations depend on the parameters, while others remain parameter-independent. Parameter-dependent operations can introduce additional partitioning, resulting in the \bnb search tree being explored differently across various parts of the parameter space.

The following definitions are useful in what follows~\cite{borrelli2017predictive}.	
\begin{definition} \label{def:partition:uni}		
A collection of sets $\{\Theta^i\}_{i=1}^N$ is
a \textit{partition} of a set $\Theta$ if (1) $ \mathring{\Theta}^i \cap \mathring{\Theta}^j = \emptyset, i\neq j$, and (2) $\cup_{i=1}^{N} \Theta^i = \Theta$, where  $\mathring{\Theta}^i$ denotes the interior of $\Theta^i$. Moreover, $\{\Theta^i\}_{i=1}^N$ is a \textit{polyhedral partition} of a polyhedral set $\Theta$ if $\{\Theta^i\}_{i=1}^N$ is a partition of $\Theta$ and $\Theta^i$ is a polyhedron for all $i$.  
\end{definition} 

\begin{definition} \label{def:pwa:uni}			
A function $h(\theta): \Theta \rightarrow  \mathbb{R}^n$, where $\Theta \subset  \mathbb{R}^{n_{\theta}}$, 
is called \textit{piecewise quadratic} (PWQ) if there exists a partition $\{\Theta^i\}_{i}$ of $\Theta$ such that $h(\theta) = \theta^T Q^i \theta + R^i \theta + S^i,$ $\forall  \theta \in  \Theta^i$, and $\forall i$. If $\{\Theta^i\}_{i}$ forms a polyhedral partition, then 
$h(\theta)$ is referred to as polyhedral PWQ (PPWQ). \\
Additionally, the function $h(\theta)$ is called \textit{piecewise affine} (PWA) if $Q^i= \mathbf{0}$, $\forall i$, and  \textit{piecewise constant} (PWC) if $Q^i= \mathbf{0}$ and $R^i= \mathbf{0}$, $\forall i$. 
Polyhedral PWA (PPWA) and polyhedral PWC (PPWC) functions are defined analogously. 
\end{definition} 

\subsection{Algorithm overview}
Algorithm~\ref{alg:cert_uni} presents a unified complexity-certification framework for \bnb algorithms, applicable to mp-MILP~\eqref{eq:mpMILP_uni} and mp-MIQP problems~\eqref{eq:mpMIQP_uni}.  Conceptually, it can be viewed as running Algorithm~\ref{alg:BnB_on_uni} for all parameters while partitioning the initial parameter set $\Theta_0$ into regions in which the solver state remains constant. 
To manage candidate and terminated regions, the algorithm maintains two lists: 
\begin{itemize}
    \item $\mathcal{S}$: the \textit{candidate list}, which stores regions from $\Theta_0$ that are yet to terminate.
    \item $\mathcal{F}$: the \textit{final list}, which contains terminated regions along with their associated complexity measures.
\end{itemize}

Each tuple $\left(\Theta, \mathcal{T}, \kappa_{\text{tot}}, \bar{J}\right)$ 
in $\mathcal{S}$ consists of:  the corresponding parameter set \( \Theta \), the sorted list of pending nodes \( \mathcal{T} \)  forming the local \bnb tree within the region, the accumulated complexity measure \( \kappa_{\text{tot}} \) up to the current region's state, and the upper bound \( \bar{J}(\theta) \) across all \( \theta \in \Theta \). Since each node here contains an mp-LP/mp-QP relaxation and is therefore parameter-dependent, it is denoted \( \eta(\theta) \).


\begin{algorithm}[htbp]
\caption{\textsc{B\&BCert}: Unified complexity-certification framework for B\&B} 
\label{alg:cert_uni}	
\begin{algorithmic}[1] 
\Require\Longunderstack[l]{mp-MILP/mp-MIQP problem~\eqref{eq:mpMILP_uni}/\eqref{eq:mpMIQP_uni} and $\Theta_0$} 
\vspace{.04cm}
\Ensure  $\mathcal{F}=\{(\Theta^i,\kappa_{tot}^i,\bar {J}^i)\}_{i}$ 
\vspace{.06cm}
\State {$\mathcal{F} \gets \emptyset$}
\State {Push $( \Theta_0, \{ (\emptyset,\emptyset)\}, 0, \infty)$ to $\mathcal{S}$} 
\label{step:init_S_uni} 
\vspace{.05cm}
\While {$\mathcal{S} \neq \emptyset$} \label{step:while_uni}
\State Pop $( \Theta, \mathcal{T}, \kappa_{{tot}}, \bar{J})$ from $\mathcal{S}$
\label{step:pop_reg_uni}
\If{$\mathcal{T} \neq \emptyset$}  
\label{step:T_empty_uni} 
\State{$ \eta(\theta)  \leftarrow$ pop node from $\mathcal{T}$}
\label{step:pop_node_cert_uni} 
\State \Longunderstack[l]{$\{(\Theta^j, \kappa^j,  \mathcal{A}^j, \barbelow{J}^{j}, \barbelow{x}^{j})\}_{j=1}^N \leftarrow$ $\textsc{solveCert} \left(\eta(\theta),\Theta \right)$} \label{step:cert_node_uni} 
\vspace{.08cm}
\For {$j \in \mathbb{N}_{1:N} $} \label{step:for_reg_cert_uni} 
\State $\kappa_{tot}^j \leftarrow \kappa_{tot} + \kappa^j $ \label{step:gener_reg_uni}
\State $\mathcal{S} \leftarrow$ $ \textsc{cutCert}\left(( \Theta^j, \mathcal{T}, \kappa_{tot}^j, \bar{J}), \barbelow{J}^j, \barbelow{x}^j,  \mathcal{A}^j, \eta, \mathcal{S} \right)$ \label{step:eval_cut_uni} 
\vspace{-0.3cm}
\EndFor		
\label{step:pop_cert_uni} 
\Else
\State{Add $( \Theta, \kappa_{{tot}}, \bar{J})$ to $\mathcal{F}$} 
\label{step:append_uni}
\EndIf
\EndWhile
\State \textbf{return} $\mathcal{F}$ 
\end{algorithmic}
\end{algorithm}

Algorithm~\ref{alg:cert_uni} begins by initializing the list $\mathcal{S}$ with the initial region (the entire parameter set $\Theta_0$) and the root node (Step~\ref{step:init_S_uni}). The algorithm then iterates over the regions stored in $\mathcal{S}$, retrieving a tuple $(\Theta, \mathcal{T}, \kappa_{{tot}}, \bar{J})$ at each step.
If $\mathcal{T}$ is empty, the \bnb tree for $\Theta$ has been fully explored, and the region is added to $\mathcal{F}$ (Step~\ref{step:append_uni}). Otherwise, the first pending node from $\mathcal{T}$ 
is selected (Step~\ref{step:pop_node_cert_uni}) and processed using the \textsc{solveCert} subroutine (Step~\ref{step:cert_node_uni}). This subroutine certifies the relaxation and partitions $\Theta$ accordingly, returning for each resulting subregion $\Theta^j$: the value function (lower bound) $\barbelow{J}^j(\theta)$, the solution $\barbelow{x}^j(\theta)$, the active set $\mathcal{A}^j$, and the complexity measure $\kappa^j$. An example of such a subroutine is provided in~\cite{arnstrom2021unifying}.

Following this, the region $\Theta$ is decomposed into $N$ subregions, each $\Theta^j$ assigned an updated accumulated complexity measure (Step~\ref{step:gener_reg_uni}) along with a copy of $\mathcal{T}$ and $\bar{J}$. An illustration of this decomposition is provided in Fig.~\ref{fig:part_uni}. Each $\Theta^j$ then undergoes the \textsc{cutCert} procedure, which applies the parameter-dependent \bnb cut conditions  (Step~\ref{step:eval_cut_uni}). This step may cut or branch the node, further partition \( \Theta^j \), and update $\mathcal{S}$ accordingly. Details on this procedure are provided in Section~\ref{subsec:cut_cund_cert_uni}.

\begin{figure}[htbp] 
	\centerline{\includegraphics[scale=0.26]{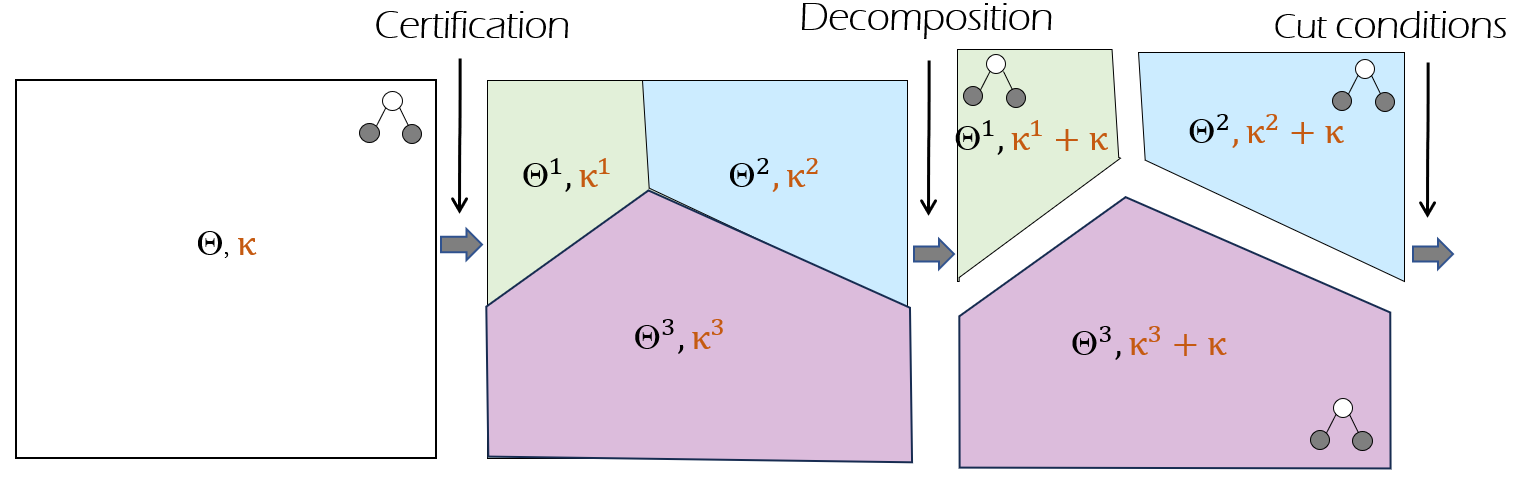}} 
	\caption{Decomposition of a polyhedral region $\Theta$ after relaxation certification using \textsc{solveCert} in Algorithm~\ref{alg:cert_uni}. Each subregion is potentially further decomposed by \textsc{cutCert}.} 
\label{fig:part_uni}
\end{figure}

The following assumptions on the certification function for subproblems are useful. 
\begin{assumption}
\label{ass:cert_uni}
The \textsc{solveCert} function satisfies the following properties. The conditions in parentheses specifically apply when the input parameter set $\Theta$ is polyhedral. 
\begin{enumerate}
\item The function partitions (the polyhedral set) $\Theta$ into a finite (polyhedral) partition $\{\Theta^j\}_{j=1}^N$. 
\item The solution $\barbelow{x}(\theta)$ is (polyhedral) PWA. 
\item The value function $\barbelow{J}(\theta)$ is (polyhedral) PWA for mp-LPs and (polyhedral) PWQ for mp-QPs. 
Infeasibility is represented by $\barbelow{J}(\theta) = \infty$.
\item The complexity measure $\kappa(\theta)$ is (polyhedral) PWC. 
\item The complexity measure obtained from \textsc{solveCert} coincides with the complexity of the online solver. 
\end{enumerate}
\end{assumption}
\begin{remark}
Any LP/QP solver with a corresponding certification method that satisfies Assumption~\ref{ass:cert_uni} (e.g.,~\cite{arnstrom2021unifying}) can be used as the \textsc{solveCert} function at Step~\ref{step:cert_node_uni} in Algorithm~\ref{alg:cert_uni}.
\end{remark}

\begin{remark}
    While the objective functions differ between MILPs and MIQPs, the fundamental principles of Algorithm~\ref{alg:cert_uni} 
    remain consistent across both problem families. 
\end{remark}

\subsection{Cut-condition evaluation} \label{subsec:cut_cund_cert_uni}
The \textsc{cutCert} procedure for  evaluating the \bnb cut conditions over a parameter set $\Theta$ is presented in Algorithm~\ref{alg:cut_cond_cert_uni}. 
This procedure begins by identifying the subset $\tilde{\Theta} \in \Theta$ in which the dominance cut condition holds 
(Step~\ref{step:cut_cond_cert1_uni}). Within this subregion the node is then cut, and the subregion is excluded from further processing in \textsc{cutCert} 
(Step~\ref{step:dom_cut_cert_uni}). Subsequently, $\Theta$ is updated to exclude $\tilde{\Theta}$ (Step~\ref{step:cut_cond_cert2_uni}). 

For the remaining part of $\Theta$, the algorithm checks whether all binary variables are fixed. 
If so, it updates the upper bound and the best solution (Step~\ref{step:int_cut_cert2_uni}) and terminates by adding the subregion to $\mathcal{S}$ (Step~\ref{step:int_cut_st_cert_uni}). Otherwise, the node is branched by selecting a relaxed binary variable (Step~\ref{step:branch1_cert_uni}), and the two new nodes are stored in $\mathcal{T}$ for further exploration  (Step~\ref{step:append_cert_uni}). 
The procedures for selecting branching indices and sorting nodes, both parameter-dependent, are detailed in Sections~\ref{subsec:branching_cert_uni} and~\ref{subsec:appending_cert_uni}, respectively. 

\begin{algorithm}[htbp]
\caption{\textsc{cutCert}: Cut-condition evaluation and branching in the certification framework} 
\label{alg:cut_cond_cert_uni}
\begin{algorithmic}[1]
\Require $\left( \Theta, \mathcal{T}, \kappa_{tot}, \bar{J} \right )$, $\barbelow{J}$, $\barbelow{x}$,  $\mathcal{A}$,  $(\mathcal{B}_0$, $\mathcal{B}_1)$, $\mathcal{S}$ 
\Ensure $\mathcal{S}$ 
\vspace{.06cm}
\State $\tilde{\Theta} \leftarrow  \{\theta \in \Theta \;|\; \barbelow{J}(\theta) \geq \bar{J}(\theta)\}$ \label{step:cut_cond_cert1_uni} 
\If {$\tilde{\Theta} \neq \emptyset$} \label{step:cut_cond_dom_uni}
\State {Push $ \left (\tilde{\Theta}, \mathcal{T}, \kappa_{\text{tot}}, \bar{J}\right )$ to $\mathcal{S}$}
\label{step:dom_cut_cert_uni}
\Comment{\textcolor{darkgray}{\textit{dominance cut}}}
\EndIf
\vspace{-0.1cm}
\State $\Theta \leftarrow \{\theta \in \Theta \;|\; \barbelow{J}(\theta) < \bar{J}(\theta)\}$ \label{step:cut_cond_cert2_uni} 
\If {$\Theta \neq \emptyset$}
\label{step:nodom_cut_cert_uni}
\If {all relaxed binary constraints are active} 
\label{step:int_cut_cert_uni}
\State $\bar{J}(\theta) \leftarrow \barbelow{J}(\theta)$, $\bar{x}(\theta) \leftarrow \barbelow{x}(\theta)$ \label{step:int_cut_cert2_uni}
\Comment{\textcolor{darkgray}{\textit{integer-feasibility cut}}}

\State {Push $ \left ( \Theta, \mathcal{T}, \kappa_{\text{tot}}, \bar{J}\right )$ to $\mathcal{S}$}
\label{step:int_cut_st_cert_uni}
\Else \label{step:cut_cond_branch_uni}
\Comment{\textcolor{darkgray}{\textit{branching}}}
\State{$\{(\Theta^k,k)\}_{k=1}^{N_b}$ $\leftarrow$ \textsc{branchIndCert}($\barbelow{x}(\theta), \mathcal{B}, \Theta$)} 
\label{step:branch1_cert_uni}
\vspace{0.04cm} 
\For {$k \in \mathbb{N}_{1:N_b} $} %
\vspace{0.04cm} 
\State \Longunderstack[l]{$\mathcal{S} \leftarrow$ \textsc{sortCert}$((\mathcal{B}_0 \cup \{k\},\mathcal{B}_1 ), (\mathcal{B}_0,\mathcal{B}_1 \cup \{k\})$, \\$(\Theta^k, \mathcal{T}, \kappa_{\text{tot}}, \bar{J}), \barbelow{J}(\theta), \mathcal{S} )$} 
\label{step:append_cert_uni}  
\EndFor
\EndIf
\EndIf
\State \textbf{return} $\mathcal{S}$
\end{algorithmic}
\end{algorithm}

When comparing the value functions  at Steps~\ref{step:cut_cond_cert1_uni} and~\ref{step:cut_cond_cert2_uni} in Algorithm~\ref{alg:cut_cond_cert_uni}, the resulting partitioning of $\Theta$ depends on the type of these functions:
\begin{itemize}[leftmargin=*]
    \item For MILPs, $\barbelow{J}(\theta)$ and $\bar{J}(\theta)$ are affine in $\Theta$. Consequently, $\Theta$ is decomposed into two polyhedral parameter sets.
    \item For MIQPs,  $\barbelow{J}(\theta)$ and $\bar{J}(\theta)$ are quadratic in $\Theta$. As a result, this partitioning introduces quadratic cuts to $\Theta$. Let 
$\barbelow{J}(\theta) = \theta^T \barbelow{Q} \theta + \barbelow{R} \theta + \barbelow{S}$ and $\bar{J}(\theta) = \theta^T \bar{Q} \theta + \bar{R} \theta + \bar{S}$ within $\Theta$. 
Their difference is 
\begin{equation} \label{eq:eq_obj_diff_qp_uni}
\tilde{J}(\theta ) = \barbelow{J}(\theta ) - \bar{J}(\theta ) = \theta^T \tilde{Q} \theta + \tilde{R} \theta + \tilde{S}, \hspace{.2cm} \forall \theta \in \Theta,
\end{equation}%
where $\tilde{Q} = \barbelow{Q} - \bar{Q}$, $\tilde{R} = \barbelow{R} - \bar{R}$, and $\tilde{S} = \barbelow{S} - \bar{S}$. The function $\tilde{J}(\theta)$ is generally an indefinite (and therefore non-convex) quadratic function. 
 To verify whether \(\tilde{J}(\theta) \geq 0\) over \(\Theta\), one can compute its minimum value by solving the following potentially indefinite QP: 

\begin{equation} \label{eq:QP_compare_uni}
\tilde{J}^* = \min_{\theta \in \Theta} \tilde{J}(\theta),
\end{equation}
where $\tilde{J}^* \geq 0$ confirms that $\barbelow{J}(\theta) \geq \bar{J}(\theta)$ throughout $\Theta$. Such an indefinite QP can be solved using, e.g., Gurobi~\cite{gurobi}. 
If $\tilde{J}^* < 0$, then $\Theta$ is further  partitioned based on the quadratic condition $\tilde{J}(\theta) \geq 0$. This introduces nonlinear inequalities in the definition of the partition, resulting in a non-polyhedral partitioning of $\Theta$. 
\end{itemize}

\subsection{Branching strategies} \label{subsec:branching_cert_uni}
Similar to the online \bnb, various branching strategies can be incorporated into the certification framework. Since the relaxation's solution \( \barbelow{x}(\theta) \) is an affine function of $\theta$ in a region \( \Theta \) for both MILPs and MIQPs, i.e., there exist matrices \( \barbelow{F} \) and \( \barbelow{g} \) such that \( \barbelow{x}(\theta) = \barbelow{F} \theta + \barbelow{g} \), the score function \( s(\theta) \) is also generally parameter-dependent. Thus, comparing score functions within \( \Theta \) results in further partitioning of \( \Theta \).


The \textsc{branchIndCert} procedure for selecting the branching index in the certification framework, one of the contributions of this paper, is detailed in Algorithm~\ref{alg:branching_cert_uni}. This algorithm corresponds to 
Algorithm~\ref{alg:branching_on_uni}, which accommodates parameter-dependent branching-index selection over \(\Theta\). It begins by identifying the set of branching candidates \(\bar{B}\) at Step~\ref{step:B_score_cert}. 
It then computes the score function $s_i(\theta)$ for each $i \in \bar{\mathcal{B}}$ (using, e.g., Algorithm~\ref{alg:mostinfeas_cert_uni}). 
At Step~\ref{step:branch_partition_cert_uni}, the algorithm identifies a subset $\Theta^k \subseteq \Theta$ where the candidate binary variable $k$ achieves the highest score $s_k(\theta)$ within $\Theta^k$. It then stores this subset in the output list \( \mathcal{F}_b \) at Step~\ref{step:branch_pushF_cert_uni}. 
The partitioning of $ \Theta $ at Step~\ref{step:branch_partition_cert_uni} is performed using hyperplanes if the score functions can be expressed in an affine form, i.e., $ s_i(\theta) \triangleq C_i \theta + d_i$, $\forall i \in \bar{\mathcal{B}}$, where $ C_i \in \mathbb{R}^{1 \times n_{\theta}} $ and $ d_i \in \mathbb{R} $.

\begin{algorithm}[htbp]
\caption{\textsc{branchIndCert}: Branching-index selection in the certification framework} 
\label{alg:branching_cert_uni}
\begin{algorithmic}[1]
\Require $\barbelow{x}(\theta)$, $\mathcal{B}$, $\Theta$
\Ensure $\mathcal{F}_b=\{(\Theta^k,k )\}_k$ 
\vspace{.06cm}
\State $\bar{\mathcal{B}} \leftarrow  $ $\{k \in \mathcal{B} \hspace{.1cm} | \hspace{.1cm} \barbelow{x}_k \notin \{0,1\} \}$
\label{step:B_score_cert}
\State{\textbf{for} $i \in \bar{\mathcal{B}}$, compute $s_i(\theta)$} 
\label{step:def_score_cert_uni}	
\For {$k$ in $\bar{\mathcal{B}}$} \label{step:for_b_bar_cert_uni}%
\State {$\Theta^{k} \leftarrow \{\theta \in \Theta \;|\; s_k(\theta) \geq s_i(\theta), \; \forall i \in \bar{\mathcal{B}} \setminus \{k\} \}$} 
\label{step:branch_partition_cert_uni}
\If {$\Theta^{k} \neq \emptyset$} 
\State Push $(\Theta^{k}, k)$ to $\mathcal{F}_b$ 
\label{step:branch_pushF_cert_uni}
\EndIf		
\EndFor	
\State \textbf{return} $\mathcal{F}_b$ \label{Step:returnR_uni}				
\end{algorithmic}
\end{algorithm}


\subsubsection*{Predetermined branching strategies}
When relaxed binary variables are selected in a predetermined order, the score function becomes parameter-independent, i.e., \( C_i = \mathbf{0} \) for all \( i \in \bar{\mathcal{B}} \). In this case, Algorithm~\ref{alg:branching_cert_uni} does not further partition \( \Theta \) but reduces to Algorithm~\ref{alg:branching_on_uni}, selecting the index that gives the highest score over the entire set \( \Theta \).

\subsubsection*{Most infeasible branching strategy}
Analogous to~\eqref{eq:branch_minf_on_uni},
the parameter-dependent score function for MIB is given by
\begin{equation} \label{eq:branch_minf_cert_uni}
    s_i(\theta) = 0.5 - |\barbelow{x}_i(\theta) - 0.5|, \quad \forall i \in \bar{\mathcal{B}}.
\end{equation}
To express \eqref{eq:branch_minf_cert_uni} explicitly as a PWA function compatible with the rest of the framework, it has to be reformulated. This can be done by first partitioning $\Theta$ along the hyperplane $\barbelow{x}_i(\theta) = \barbelow{F}_i \theta + \barbelow{g}_i = 0.5$, resulting in two subsets:
$ \Theta^{i_1} = \{\theta \in \Theta \mid \barbelow{F}_i \theta + \barbelow{g}_i \leq 0.5\}$ and $\Theta^{i_2} = \{\theta \in \Theta \mid \barbelow{F}_i \theta + \barbelow{g}_i > 0.5\}$, $\forall i \in \bar{\mathcal{B}}$.
Within each subset, \( s_i(\theta) \) is affine: 
\begin{align} \label{eq:minf_score_cert_uni} s_i(\theta) =  \begin{cases}   \barbelow{F}_i \theta + \barbelow{g}_i,  & \forall \theta \in \Theta^{i_1}
\\[0.06cm]  -\barbelow{F}_i \theta - \barbelow{g}_i + 1, 
 & \forall \theta \in \Theta^{i_2} 
 \end{cases}\end{align} 

The \textsc{mostInfScoreCert} procedure, which computes the affine score function for MIB, is presented in Algorithm~\ref{alg:mostinfeas_cert_uni}. 
The algorithm maintains two lists, \( \mathcal{S}_b \) and \( \mathcal{F}_b \), to manage the candidate and final regions, respectively. 
At each iteration, the algorithm extracts a tuple \( (\Theta, i_b, C, d) \) from \( \mathcal{S}_b \), where \( \Theta \) is the current parameter set, \( i_b \) is an index used to iterate over the elements of the candidate set \( \bar{\mathcal{B}} \), and \( s(\theta) = C \theta + d \) is the corresponding affine score function. If \( i_b \) is within the range of \( \bar{\mathcal{B}} \), the element \( i = \bar{\mathcal{B}}_{i_b} \) (the \( i_b \)-th element of \( \bar{\mathcal{B}} \)) is selected (Step~\ref{step:minf_cert_j_uni}) and the parameter set is partitioned into two subsets 
(Step~\ref{step:minfb_partitioned_cert}). For each non-empty subset, the corresponding affine score function is computed using~\eqref{eq:minf_score_cert_uni} (Steps~\ref{step:minfb_Theta1_cert} and~\ref{step:minfb_Theta2_cert}). The new subsets are pushed back into \( \mathcal{S}_b \) with the incremented index \( i_b + 1 \) (Steps~\ref{step:minf_cert_pushS1} and~\ref{step:minf_cert_pushS2}). If the index $i_b$ exceeds the size of \( \bar{\mathcal{B}} \) (indicating that all elements in \( \bar{\mathcal{B}} \) have been processed within the region), the parameter set and score function are added to \( \mathcal{F}_b \) (Step~\ref{step:minfb_pushF_cert}). This process results in a partition of \( \Theta \) with the corresponding PWA score functions. The output of this procedure is used in Step~\ref{step:def_score_cert_uni} of Algorithm~\ref{alg:branching_cert_uni} to compute the MIB score function.

\begin{algorithm}[htbp]    \caption{\textsc{mostInfScoreCert}: Most infeasible branching score for the certification framework} 
\label{alg:mostinfeas_cert_uni}
\begin{algorithmic}[1]
\Require 
$\barbelow{x} = \barbelow{F} \theta + \barbelow{g}$, $\bar{\mathcal{B}}$, $\Theta$ 
\Ensure $\mathcal{F}_b = \{ (\Theta^j, C^j, d^j )\}_j$ 
\vspace{.06cm}
\State Push $(\Theta, 1, \texttt{NaN}, \texttt{NaN})$ to $\mathcal{S}_b$
\While{$\mathcal{S}_b \neq \emptyset$}
\State Pop $(\Theta, i_b, C, d)$ from $\mathcal{S}_b$ \label{step:minfb_popS_cert}
\If{$i_b \leq |\bar{\mathcal{B}}|$ }
\State{$i \leftarrow \bar{\mathcal{B}}_{i_b}$}
\label{step:minf_cert_j_uni}
\State \Longunderstack[l]{$\Theta^{i_1} = \{\theta \in \Theta \hspace{.1cm} | \hspace{.1cm}  \barbelow{x}_i (\theta) \leq 0.5 \} $  \\ $\Theta^{i_2} = \{\theta \in \Theta \hspace{.1cm} | \hspace{.1cm}  \barbelow{x}_i (\theta) > 0.5 \}$} 
\label{step:minfb_partitioned_cert}
\If {$\Theta^{i_1} \neq \emptyset$} 
\State{$C_i \leftarrow \barbelow{F}_i$, $d_i \leftarrow \barbelow{g}_i$} \label{step:minfb_Theta1_cert}
\State Push $(\Theta^{i_2}, i_b+1, C, d)$ to $\mathcal{S}_b$ \label{step:minf_cert_pushS1}
\EndIf
\If {$\Theta^{i_2} \neq \emptyset$}
\State{$C_i \leftarrow -\barbelow{F}_i$, $d_i \leftarrow -\barbelow{g}_i + 1$} \label{step:minfb_Theta2_cert}
\State Push $(\Theta^{i_2}, i_b+1, C, d)$ to $\mathcal{S}_b$ \label{step:minf_cert_pushS2}
\EndIf
\Else
\State Push $(\Theta, C, d)$ to $\mathcal{F}_b$ \label{step:minfb_pushF_cert}
\EndIf
\EndWhile
\State \textbf{return} $\mathcal{F}_b$ \label{step:minfb_return_cert_uni}		
\end{algorithmic}
\end{algorithm}


\subsection{Node-selection strategies}
\label{subsec:appending_cert_uni}
Similar to the online \textsc{B\&B}, various node-selection strategies can be incorporated into the certification framework. Since the sorting criterion \( \rho(\eta(\theta)) \) is generally parameter-dependent, sorting nodes in \( \mathcal{T} \) based on \( \rho(\eta(\theta)) \) within a region $\Theta$ can further partition $\Theta$.  In~\cite{shoja2022overall}, DF was used to certify the computational complexities of MIQPs, while~\cite{shoja2022exact} explored different node-selection strategies for MILPs. In this work, we extend the results of~\cite{shoja2022exact} to the unified certification framework in Algorithm~\ref{alg:cert_uni}.

The \textsc{sortCert} procedure, which appends new nodes to $\mathcal{T}$ based on the node-selection strategy within the certification framework, is detailed in Algorithm~\ref{alg:append_cert_uni}. This algorithm corresponds to 
Algorithm~\ref{alg:append_on_uni}, which accommodates parameter-dependent node sorting over $\Theta$. %
At Step~\ref{Step:append_compare_cert_uni1}, the sorting criterion of the new nodes is compared with those of existing nodes in $\mathcal{T}$, identifying subsets of $\Theta$ where they have higher priority. The new nodes are stored in the list within this region, which is subsequently added to 
$\mathcal{S}$ at Step~\ref{step:append_off_push_uni}.

\begin{algorithm}[htbp]
\caption{\textsc{sortCert}: Sorted node insertion 
in the certification framework} 
\label{alg:append_cert_uni}
\begin{algorithmic}[1]
\Require $\tilde{\eta}_{0}(\theta), \tilde{\eta}_{1}(\theta)$,  $(\Theta$, 
 $\mathcal{T} = \{\eta_i\}_{i=1}^{N_{\mathcal{T}}}$ , $\kappa_{tot}$, $\bar{J})$, $\barbelow{J}(\theta)$, $\mathcal{S}$  
\Ensure  $\mathcal{S}$ 
\vspace{.06cm}
\State Compute $\rho(\tilde{\eta}_0(\theta))$ based on the node-selection strategy
\For {$i \in \mathbb{N}_{1:N_{\mathcal{T}}}$} 
\label{Step:sort_off_for_u} 
\State {$\Theta^i \leftarrow \{\theta \in \Theta \;|\; \rho(\tilde{\eta}_0(\theta))  \leq  \rho(\eta_i(\theta)) \} $} \label{Step:append_compare_cert_uni1} 
\If {$\Theta^i \neq \emptyset$}
\State{$\mathcal{T}^i = \{ \{\eta_j\}_{j=1}^{i-1}, \tilde{\eta}_{0}, \tilde{\eta}_{1},\{\eta_j\}_{j=i}^{N_{\mathcal{T}}}\}$}\label{step:append_off_updateT_uni}
\State Push $\left ( \Theta^i,\mathcal{T}^i,\kappa_{tot},\bar{J} \right )$ to $\mathcal{S}$ \label{step:append_off_push_uni}
\State {$\Theta \leftarrow \{\theta \in \Theta \;|\; \rho(\tilde{\eta}_0(\theta))  >  \rho(\eta_i(\theta)) \} $} \label{Step:append_compare_cert_uni2} 
\EndIf
\If {$\Theta = \emptyset$} 	\label{step:sort_off_Th_empty_uni}	
\State \textbf{return} $\mathcal{S}$ 
\EndIf	
\EndFor	
\end{algorithmic}
\end{algorithm} 

The sorting criteria for DF and BrF depend solely on the node's level (see Section~\ref{subsec:node_selec_on_uni}) and are, hence, parameter-independent. As a result, at Step~\ref{Step:append_compare_cert_uni1}, \( \Theta^i \) is either the entire set \( \Theta \) or empty. Thus, no additional partitioning of \( \Theta \) occurs, and Algorithm~\ref{alg:append_cert_uni} reduces to  Algorithm~\ref{alg:append_on_uni} for both strategies.  
In contrast, BF employs a parameter-dependent sorting criterion with \( \rho_{\text{BF}}(\eta(\theta)) = \barbelow{J}(\theta) \). Consequently, Algorithm~\ref{alg:append_cert_uni} further partitions \( \Theta \) based on the type of \( \barbelow{J}(\theta) \), using hyperplanes for MILPs and quadratic functions for MIQPs.

\subsection{Complexity certification of heuristics in
B\&B} \label{sec:heuris_uni}
In this section, we investigate how the primal heuristics reviewed in Section~\ref{sec:heuris_on_uni} can be certified and, hence, integrated into the certification framework. To this end, parameter-dependent versions of these heuristics are required.  The results herein not only preserve the theoretical rigor of the framework but also improve the practical efficiency of online \bnb algorithms by tailoring the heuristic choices to the specific problem. 

\subsubsection{Start Heuristics}  
\label{subsec:st_heu_cert_uni}  
In~\cite{shoja2023sheuristic}, parameter-dependent versions of three common start heuristics—RENS, diving, and (objective) feasibility pump—were presented for mp-MILPs. These methods use the solution \( \barbelow{x}(\theta) \) from the solution at the root node to search for an integer-feasible solution \( \hat{x}(\theta) \). Their integration into the complexity-certification algorithm limited to MILPs was also demonstrated in~\cite{shoja2023sheuristic}.  

Except for the objective feasibility pump, which incorporates the objective function, RENS, diving, and feasibility pump rely only on \( \barbelow{x}(\theta) \). Since \( \barbelow{x}(\theta) \) is affine for both mp-MILPs and mp-MIQPs within a region \( \Theta \), the results in~\cite{shoja2023sheuristic} can be applied similarly to mp-MIQPs. Consequently, these 
start heuristics can be integrated into the unified certification framework in Algorithm~\ref{alg:cert_uni} for both MILPs and MIQPs.

\subsubsection{Improvement heuristics} \label{subsec:imp_heu_cert_uni}
Consider a parameter set $\Theta$ and assume that the integer-feasible solution $\bar{x}(\theta)$ has been found within $\Theta$. Since $\bar{x}(\theta)$ is affine in $\Theta$, there exist matrices $\bar{F}$ and $\bar{g}$ such that $\bar{x}(\theta) = \bar{F} \theta + \bar{g}$. To demonstrate that the certification framework can also accommodate improvement heuristics, this paper presents parameter-dependent versions of LB and RINS. 
Analogous to the online \bnb, we restrict these heuristics to mp-MILPs (see Section~\ref{subsec:imp_heu_on_heu}). 

\vspace{.1cm} 
\textit{Certification of the local branching improvement heuristic} \label{subsubsec:LB_heuris_cert_uni}\hspace{.1cm}
The \textsc{LBCert} procedure in Algorithm~\ref{alg:heuris_LB_cert_uni} presents the complexity certification of the LB heuristic.  
The algorithm maintains two lists, \( \mathcal{S}_h \) and \( \mathcal{F}_h \), to manage the candidate and final regions, respectively. Each tuple \( (\Theta, \hat{\kappa}^h, \bar{x}(\theta),  r_n, \texttt{exec}) \) in \( \mathcal{S}_h \) consists of: the current parameter set \( \Theta \), the heuristic accumulated complexity measure \( \hat{\kappa}^h \), the solution \( \bar{x}(\theta) \), the neighborhood size $r_n$, and a boolean variable \( \texttt{exec} \) indicating whether execution should continue. At each iteration, a sub-mpMILP problem $\hat{\mathcal{P}}_{\text{mpLB}}(\theta)$ from the original problem~\eqref{eq:mpMILP_uni} is formulated (Step~\ref{step:LB_cert_prob}) and certified using \textsc{B\&BCert}  in Algorithm~\ref{alg:cert_uni} (Step~\ref{step:LB_cert_uni}).  These subproblems are additionally constrained with the $r_n$-neighborhood constraints~\eqref{eq:LB_neigh_on_uni} and the parameter-dependent  objective cut-off constraints in~\eqref{eq:LB_cutoff_on_uni} 
given by
\begin{equation} \label{eq:LB_cutoff_cert_uni}
c^T x(\theta) \leq (1-\varepsilon) c^T \bar{F} \theta + (1-\varepsilon) c^T \bar{g}, \quad \varepsilon > 0.
\end{equation}

\begin{algorithm}[htbp]
\caption{\textsc{LBCert}: Local branching heuristic certification} 
\label{alg:heuris_LB_cert_uni}
\begin{algorithmic}[1]
\Require {mp-MILP problem~\eqref{eq:mpMILP_uni}, $r_{n_0}$, \(\bar{x}(\theta) = \bar{F} \theta + \bar{g}\), \(\Theta\)} 
\Ensure \(\mathcal{F}_h = \{(\Theta^j, \hat{\kappa}^{h_j}, \hat{x}^j(\theta))\}_j\) 
\vspace{.06cm}
\State Push $(\Theta, 0, \bar{x}(\theta), r_{n_0}, \texttt{true})$ to $\mathcal{S}_h$
\label{step:LB_start_cert}
\While{\(\mathcal{S}_h \neq \emptyset\)} 
\State Pop \((\Theta, \hat{\kappa}^h, \bar{x}(\theta),  r_{n}, \texttt{exec})\) from \(\mathcal{S}_h\)
\label{step:LB_pop_cert}
\If{\(\texttt{exec}\)} \label{step:LB_cert_ifexec}
\State \Longunderstack[l]{Formulate subproblem \(\hat{\mathcal{P}}_{\text{mpLB}}(\theta)\) by adding constraints \\ ~\eqref{eq:LB_neigh_on_uni} and~\eqref{eq:LB_cutoff_cert_uni} to the original  problem~\eqref{eq:mpMILP_uni}}
\label{step:LB_cert_prob}
\vspace{.01cm}
\State \Longunderstack[l]{\(\{(\Theta^i, \hat{\kappa}^i, \hat{J}^i(\theta), \hat{x}^i(\theta))\}_{i=1}^N \gets \textsc{B\&BCert}(\hat{\mathcal{P}}_{\text{mpLB}}, \Theta)\)}
\label{step:LB_cert_uni}
\For{$i \in \mathbb{N}_{1:N}$} \label{step:LB_for_cert}
\If{\(\hat{x}^i(\theta)  \neq \bar{x}(\theta)\) in \(\Theta^i\)} \label{step:LB_cert_if}
\Comment\textcolor{darkgray}{$\hat{F}^i \neq \bar{F}$ and $\hat{g}^i \neq \bar{g}$} 
\State Push \((\Theta^i, \hat{\kappa}^{h} + \hat{\kappa}^i, \hat{x}^i(\theta))\) to \(\mathcal{F}_h\)
\label{step:LB_cert_F}
\Comment\textcolor{darkgray}{\textit{{new solution}}}
\Else
\State \Longunderstack[l]{$\texttt{exec},r_{n} \gets\textsc{neighborSize}(\texttt{exec},r_{n},r_{n_0},\hat{J}^i(\theta))$}\label{step:LB_kneigh_cert}
\State Push \((\Theta^i, \hat{\kappa}^{h} + \hat{\kappa}^i, \bar{x}(\theta), r_{n}, \texttt{exec})\) to \(\mathcal{S}_h\)
\label{step:LB_cert_S} 
\EndIf
\EndFor
\Else
\State Push \((\Theta, \hat{\kappa}^{h}, \texttt{NaN})\) to \(\mathcal{F}_h\)
\label{step:LB_push_cert}
\EndIf
\EndWhile   
\State \Return \(\mathcal{F}_h\)  \label{Step:LB_cert_return}
\end{algorithmic}
\end{algorithm}

If a better integer-feasible solution $\hat{x}^i(\theta) = \hat{F}^i \theta + \hat{g}^i$ ($\neq \bar{x}(\theta)$) is found within a region $\Theta^i$ (Step~\ref{step:LB_cert_if}), it is stored in $\mathcal{F}_h$ (Step~\ref{step:LB_cert_F}). Otherwise, the neighborhood size $r_n$ is adjusted (Step~\ref{step:LB_kneigh_cert}) using the \textsc{neighborSize} procedure from Algorithm~\ref{alg:heuris_LB_on_uni}, and the region is requeued for further exploration.
The algorithm ultimately returns a partition $\{\Theta^j\}_j$, where each $\Theta^j$ contains a potentially found improved solution $\hat{x}^j(\theta)$ and its associated complexity measure $\hat{\kappa}^{h_j}$.

\vspace{.1cm} 
\textit{Certification of the RINS improvement heuristic} \label{subsubsec:RINS_heuris_cert_uni}\hspace{.2cm}  
The \textsc{RINSCert} procedure in Algorithm~\ref{alg:heuris_rins_cert_uni} presents the complexity certification of the RINS heuristic. In this algorithm, a sub-mpMILP $ \hat{\mathcal{P}}_{\text{mpRINS}}(\theta) $ from the original problem~\eqref{eq:mpMILP_uni} is formulated (Step~\ref{step:RINS_cert_prob}) and certified using \textsc{B\&BCert}  in Algorithm~\ref{alg:cert_uni} (Step~\ref{step:rins_cert_uni}). 
This subproblem is additionally constrained by fixing variables where \( \bar{x}_i(\theta) = \barbelow{x}_i(\theta) \) for all \( i \in \mathcal{B} \) (see~\eqref{eq:Rins_add_cons_heu}) while incorporating the objective cut-off constraint~\eqref{eq:LB_cutoff_cert_uni}. Similar to the online RINS in Algorithm~\ref{alg:rins_on_uni}, this heuristic is applied only if a sufficient fraction \( r \in [0,1] \) of the relaxed binary variables can be fixed (Step~\ref{step:rins_on_if_heu}). 

		\begin{algorithm}[htbp]			\caption{\textsc{RINSCert}: RINS heuristic certification} 
		\label{alg:heuris_rins_cert_uni}
		\begin{algorithmic}[1]
        \Require {mp-MILP problem~\eqref{eq:mpMILP_uni}, $\barbelow{x} = \barbelow{F} \theta + \barbelow{g}$, $\bar{x} = \bar{F} \theta + \bar{g}$, $\Theta$} 
        \Ensure $\mathcal{F}_h = \{(\Theta^j, \hat{\kappa}^{h_j}, \hat{x}^j(\theta))\}_j$ 
        \vspace{.06cm}
        		\State $\mathcal{F}_h \gets \emptyset$, $\hat{x} \gets \texttt{NaN}$, $\hat{\kappa}^{h} \gets 0$ 
                \State{$\hat{\mathcal{B}} \leftarrow \{i \in \mathcal{B} \hspace{.1cm} | \hspace{.1cm} \bar{F}_i = \barbelow{F}_i \hspace{.005cm}, \bar{g}_i = \barbelow{g}_i \hspace{.005cm}\}$} \label{step:rins_cert_B_heu}
        \If{$|\hat{\mathcal{B}}| \geq r |\mathcal{B}|$} \label{step:rins_cert_if_heu} 
			
            \State \Longunderstack[l]{Formulate subproblem $\hat{\mathcal{P}}_{\text{mpRINS}}(\theta)$ by adding constraints\\ ~\eqref{eq:Rins_add_cons_heu} and~\eqref{eq:LB_cutoff_cert_uni} to the original  problem~\eqref{eq:mpMILP_uni}} 
            \label{step:RINS_cert_prob}
			\vspace{.02cm}
            \State 
            $\mathcal{F}_h \gets \textsc{B\&BCert}(\hat{\mathcal{P}}_{\text{mpRINS}}, \Theta)$ 
\label{step:rins_cert_uni} 
			\EndIf
			\State \textbf{return} $\mathcal{F}_h$
            \label{Step:rins_cert_return_heu}
		\end{algorithmic}
	\end{algorithm}
    
\vspace{.1cm} 
\textit{Integration of heuristics into the B\&B certification framework} \label{subsubsec:LB_heu_integ_cert} \hspace{.2cm}
Similar to the online B\&B algorithm, a parameter-dependent start heuristic can be invoked in  Algorithm~\ref{alg:cert_uni} after solving the root node (i.e., between Steps~\ref{step:cert_node_uni} and~\ref{step:for_reg_cert_uni}).
Similarly, a parameter-dependent improvement heuristic can be invoked in Algorithm~\ref{alg:cert_uni} after an integer-feasible solution \( \bar{x}(\theta)\) is found for the first time (between Steps~\ref{step:int_cut_cert2_uni} and~\ref{step:int_cut_st_cert_uni} in Algorithm~\ref{alg:cut_cond_cert_uni}).  
The invoked heuristic may further partition \(\Theta\) into regions \(\{\Theta^j\}_{j=1}^{N_h}\). For each $\Theta^j$, if a better integer-feasible solution (\( \hat{x}(\theta) \neq \texttt{NaN}\)) is found, then \( \bar{x}^j(\theta) \) and \( \bar{J}^j(\theta) \) are updated to \( \hat{x}^j(\theta) \) and \( \hat{J}^j(\theta) \), respectively, within $\Theta^j$. Additionally, the complexity measure  $\hat{\kappa}^{h_j}$ associated with $\Theta^j$ is added to the accumulated complexity measure \( \kappa_{\text{tot}} \), accounting for the overall computational effort in $\Theta^j$.  
Similar to the online algorithm, these heuristics can be invoked once or multiple times during the certification process. However, to ensure correct certification, all heuristic invocations must occur identically in both Algorithms~\ref{alg:BnB_on_uni} and~\ref{alg:cert_uni}. 

\section{Properties of the unified certification framework} \label{subsec:prop_cert_uni}
In this section, we analyze the properties of the unified complexity-certification framework. Our objective is to establish that the results of Algorithm~\ref{alg:cert_uni}  coincide, pointwise, with those of the online \textsc{B\&B} Algorithm~\ref{alg:BnB_on_uni} for any parameter in $\Theta_0$, for both MILPs and MIQPs. 
To ensure meaningful results, we assume that the \bnb algorithms considered in this work satisfy the following assumption.  
\begin{assumption} \label{ass:order-tree_uni} 
The tree-exploration strategy is identical in both Algorithms~\ref{alg:BnB_on_uni} and~\ref{alg:cert_uni}. Specifically, both algorithms employ the same branching, node-selection, and heuristic strategies, with heuristic invocation performed identically in both. Additionally, the order in which the \( 0 \)-branch and \( 1 \)-branch are explored coincides in both algorithms.
\end{assumption}

To ensure the correctness of Algorithm~\ref{alg:cert_uni}, we first analyze the properties of the helper procedures invoked in the algorithm (Algorithms~\ref{alg:cut_cond_cert_uni}--\ref{alg:append_cert_uni}) over an input parameter set $\Theta$. 

\subsection{Properties of the helper procedures in the framework} \label{subsec:prop_proc_uni}
In this section, we show that the results of the parameter-dependent helper procedures coincide pointwise with their corresponding online counterparts in Algorithm~\ref{alg:BnB_on_uni}. Proofs of the following lemmas are provided in Appendix.

\begin{lemma}[Equivalence of \textsc{sortCert} and \textsc{sort}]  
\label{lem:append_uni}  
Assume that Assumptions~\ref{ass:cert_uni}--\ref{ass:order-tree_uni} hold. Then, Algorithm~\ref{alg:append_cert_uni} (\textsc{sortCert}) partitions $\Theta$ into subsets $\{\Theta^i\}_i$, such that the updated $\mathcal{T}^i$ (stored in $\mathcal{S}$) coincides with the updated $\mathcal{T}$ 
returned by Algorithm~\ref{alg:append_on_uni} (\textsc{sort}) for any fixed $\theta \in \Theta^i$, $\forall i$.  
\end{lemma}
\begin{lemma}[Equivalence of \textsc{mostInfScoreCert} and~\eqref{eq:branch_minf_on_uni}]
\label{lem:most_inf_score_uni}
Assume that Assumption~\ref{ass:cert_uni} holds. Then, Algorithm~\ref{alg:mostinfeas_cert_uni} (\textsc{mostInfScoreCert}) partitions $\Theta$ into $\{ \Theta^j\}_j$, and for each $\Theta^j$, the affine score function satisfies~\eqref{eq:branch_minf_on_uni}. That is, $s^j(\theta) = C^j \theta + d^j = 0.5 - |\barbelow{x}(\theta) - 0.5|$,  $\forall \theta \in \Theta^j$, $\forall j$.
\end{lemma}

\begin{lemma}
\label{lem:branch_uni} \textit{(Equivalence of \textsc{branchIndCert} and \textsc{branchInd})}:    
Assume that Assumptions~\ref{ass:cert_uni}--\ref{ass:order-tree_uni} hold. Then, Algorithm~\ref{alg:branching_cert_uni} (\textsc{branchIndCert}) partitions $\Theta$ into subsets $\{\Theta^k\}_k$ and returns a PWC branching index $k(\theta)$, 
such that the selected $k$ for   $\Theta^k$ 
coincides with that  selected by Algorithm~\ref{alg:branching_on_uni} (\textsc{branchInd}) for any fixed $\theta \in \Theta^k$, $\forall k$.  
\end{lemma}  
\begin{lemma}[Equivalence of \textsc{cutCert} and \textsc{cut}] 
\label{lem:cut_cond_uni}  
Assume that Assumptions~\ref{ass:cert_uni}--\ref{ass:order-tree_uni} hold. Then, Algorithm~\ref{alg:cut_cond_cert_uni} (\textsc{cutCert}) partitions $\Theta$ into subsets $\{\Theta^i\}_i$, such that the updated $\mathcal{T}^i$ and $\bar{J}^i(\theta)$ (stored in $\mathcal{S}$) coincide with the updated $\mathcal{T}$ and $\bar{J}$ returned by Algorithm~\ref{alg:cut_cond_on_uni} (\textsc{cut}) for any fixed $\theta \in \Theta^i$, $\forall i$.  
\end{lemma} 
\subsection{Properties of the core   framework} 
After establishing the one-to-one correspondence between the helper procedures in Algorithm~\ref{alg:cert_uni} and their online counterparts, we can now analyze the properties of the framework as a whole. We begin by discussing the decomposition steps, followed by the analysis of the algorithm's properties.

To summarize, 
the parameter set is partitioned at each iteration of Algorithm~\ref{alg:cert_uni} through the following steps:
\begin{enumerate}[label=(\roman*)]  
    \item \label{item:part_uni0}
    Relaxation certification using \textsc{solveCert} (Step~\ref{step:cert_node_uni})
    \item \label{item:part_uni1} 
    Evaluation of the dominance cut (Algorithm~\ref{alg:cut_cond_cert_uni}; Steps~\ref{step:cut_cond_cert1_uni} and~\ref{step:cut_cond_cert2_uni}).  
    \item \label{item:part_uni2} Selection of the branching index 
 (Algorithm~\ref{alg:branching_cert_uni}; Step~\ref{step:branch_partition_cert_uni}).  
    \item \label{item:part_uni3} Sorting and storing nodes (Algorithm~\ref{alg:append_cert_uni}; Steps~\ref{Step:append_compare_cert_uni1} and~\ref{Step:append_compare_cert_uni2}).  
\end{enumerate}  

The following lemma establishes the correctness of the partitioning performed in Algorithm~\ref{alg:cert_uni}.  

\begin{lemma}[Maintenance of complete partition]  
\label{lem:maintenance_uni}  
Assume that Assumption~\ref{ass:cert_uni} holds. At any iteration of Algorithm~\ref{alg:cert_uni}, the union of regions in $\mathcal{S}$ and $\mathcal{F}$ forms a partition of $\Theta_0$.  
\end{lemma}  

\begin{proof}
The result follows directly from Assumption~\ref{ass:cert_uni} and Lemmas~\ref{lem:append_uni}--\ref{lem:cut_cond_uni}, which establish the correctness of partitioning at each decomposition step of Algorithm~\ref{alg:cert_uni}, along with the fact that \( \mathcal{S} \) is initialized with \( \Theta_0 \) and \( \mathcal{F} \) is initially empty. A more detailed proof can be constructed following the same reasoning as the proof of Lemma~1 in~\cite{shoja2022overall}. 
\end{proof}

\begin{corollary}[Complete partition at termination]  
Assume that Assumption~\ref{ass:cert_uni} holds and that Algorithm~\ref{alg:cert_uni} has terminated with the output $\mathcal{F}$ = $\{(\Theta^i, \kappa_\text{tot}^i, \bar{J}^i)\}_i$. Then, $\{\Theta^i\}_{i}$ forms a partition of $\Theta_0$.  
\end{corollary}  

\begin{proof}  
The result follows directly from Lemma~\ref{lem:maintenance_uni} and the fact that Algorithm~\ref{alg:cert_uni} terminates when $\mathcal{S} = \emptyset$.  
\end{proof}  

The following theorems analyze the properties of the unified certification Algorithm~\ref{alg:cert_uni}.

\begin{theorem}[Equivalence of explored nodes sequences]
\label{thrm:node_seq_uni}
Assume that Assumptions~\ref{ass:cert_uni}--\ref{ass:order-tree_uni} hold. Let $\hat{\mathbb{B}}(\theta)$ denote the sequence of nodes explored to solve the 
problem~\eqref{eq:MILP_uni}/\eqref{eq:MIQP_uni} for any fixed $\theta \in \Theta_0$ using Algorithm~\ref{alg:BnB_on_uni} (\textsc{B\&B}). Moreover, let $\mathbb{B}(\theta)$ denote the sequence of nodes explored by Algorithm~\ref{alg:cert_uni} (\textsc{B\&BCert}) applied to the problem~\eqref{eq:mpMILP_uni}/\eqref{eq:mpMIQP_uni} for a terminated region $\Theta^i \ni \theta$ in $\mathcal{F}$ $= \{(\Theta^i, \kappa_{\text{tot}}^i, \bar{J}^i)\}_i$. Then, $\hat{\mathbb{B}}(\theta) = \mathbb{B}(\theta)$, $\forall \theta \in \Theta^i$, and all $i$.
\end{theorem}

\begin{proof}
Define an iteration in Algorithm~\ref{alg:cert_uni} as one execution of Steps~\ref{step:pop_reg_uni}--\ref{step:append_uni}.  
By Lemma~\ref{lem:cut_cond_uni}, the partitioning performed during any iteration ensures that the regions in $\mathcal{S}$ and $\mathcal{F}$ do not overlap and fully cover $\Theta_0$. Therefore, for any $\theta \in \Theta_0$, there exists a unique region $\Theta$ such that $\theta \in \Theta$. Furthermore, by the structure of Algorithm~\ref{alg:cert_uni}, no operations are performed outside the region $\Theta$ popped at Step~\ref{step:pop_reg_uni}.  
Thus, only iterations for which $\Theta \ni \theta$ 
need to be considered to complete the certification for $\theta$. Since all other iterations do not affect the sequence of explored nodes for this $\theta$, the role of $\mathcal{S}$ can be omitted here. Consequently, Algorithm~\ref{alg:cert_uni} can be interpreted as an iteration over the sorted node list $\mathcal{T}$, similar to the main loop started at Step~\ref{step:loop_T_on_uni} in Algorithm~\ref{alg:BnB_on_uni}.
We now proceed the proof by induction over the nodes in $\mathcal{T}$.

At the first iteration, $\mathcal{T}$ is in both algorithms initialized with the root node $(\emptyset, \emptyset)$. Thus, the sequences of explored nodes are identical at the start, 
confirming the base case.

Now, consider an arbitrary iteration of Algorithm~\ref{alg:cert_uni} and the corresponding iteration of Algorithm~\ref{alg:BnB_on_uni}. Assume that the respective node lists $\mathcal{T}$ in the two algorithms coincide at the beginning of this iteration, i.e., at Step~\ref{step:loop_T_on_uni} in Algorithm~\ref{alg:BnB_on_uni} and Step~\ref{step:T_empty_uni} in Algorithm~\ref{alg:cert_uni}. It will now be shown that $\mathcal{T}$ remains identical in both algorithms in the end of this iteration. 
In both algorithms, the first node in $\mathcal{T}$ is selected for exploration. 
The relaxation  is certified over $\Theta$ in Algorithm~\ref{alg:cert_uni} (Step~\ref{step:cert_node_uni}) and is solved for $\theta$ in Algorithm~\ref{alg:BnB_on_uni} (Step~\ref{step:solve_on_uni}).  
By Assumption~\ref{ass:cert_uni}, the computed solutions and complexity measures remain identical for $\theta$.  
After solving the relaxation, both algorithms proceed by applying B\&B cut conditions: Algorithm~\ref{alg:cert_uni} applies \textsc{cutCert} at Step~\ref{step:eval_cut_uni} while Algorithm~\ref{alg:BnB_on_uni} applies \textsc{cut} at Step~\ref{step:cutcond_on_uni}. In particular, \textsc{cutCert} partitions $\Theta$ into subsets $\{\Theta^j\}_j$, such that there is a unique region $\Theta^j$ that contains $\theta$.  Since only $\Theta^j \ni \theta$ is relevant for $\theta$, we restrict the analysis to $\Theta^j$. By Lemma~\ref{lem:cut_cond_uni}, the updated $\mathcal{T}^j$ from \textsc{cutCert} coincides with the updated $\mathcal{T}$ from \textsc{cut} for $\theta \in \Theta^j$.  
That is, the node lists remain identical in both algorithms for \( \theta \) at the end of the iteration.

Thus, by induction, the node lists remain identical in all iterations for \( \theta \). That is, the nodes contained in the respective algorithms' \( \mathcal{T} \), and the nodes' order, coincides at every iteration, ensuring that the sequence of explored nodes is identical in both algorithms, i.e., \( \hat{\mathbb{B}}(\theta) = \mathbb{B}(\theta) \) for \( \theta \). Since \( \theta \) was chosen arbitrarily, this holds for any \( \theta \in \Theta_0\), which completes the proof.
\end{proof}

\begin{theorem}
[Equivalence of complexity measures]
\label{thrm:complexity_meas_uni} 
Assume that Assumptions~\ref{ass:cert_uni}--\ref{ass:order-tree_uni} hold, and let $\kappa^*_{\text{tot}}(\theta)$ denote the accumulated complexity measure to solve the problem~\eqref{eq:MILP_uni}/\eqref{eq:MIQP_uni} problem for any fixed $\theta \in \Theta_0$ using Algorithm~\ref{alg:BnB_on_uni}.
Then, the complexity measure $\kappa_{\text{tot}}^i(\theta)$ 
returned by Algorithm~\ref{alg:cert_uni} applied to  the problem~\eqref{eq:mpMILP_uni}/\eqref{eq:mpMIQP_uni} for a terminated region $\Theta^i \ni \theta$ in $\mathcal{F} = \{(\Theta^i, \kappa^i_{\text{tot}}\bar{J}^i)\}_i$ satisfies $\kappa^i_{\text{tot}}= \kappa^*_{\text{tot}}$, $\forall \theta \in \Theta^i$, and all $i$.
\end{theorem}

\begin{proof}
By Assumption~\ref{ass:cert_uni},  \textsc{solveCert} correctly certifies the relaxations. Specifically, for any arbitrary node, the resulting complexity measure $\kappa^j(\theta)$ returned at Step~\ref{step:cert_node_uni} in Algorithm~\ref{alg:cert_uni} for a region $\Theta^j$ is identical to $\kappa$ returned at Step~\ref{step:solve_on_uni} in Algorithm~\ref{alg:BnB_on_uni} for any fixed $\theta \in \Theta^j$. Furthermore, by Theorem~\ref{thrm:node_seq_uni}, both Algorithms~\ref{alg:cert_uni} and \ref{alg:BnB_on_uni} explores the same sequence of relaxations for  $\theta$. As a result, the accumulated complexity measure $\kappa^i_{\text{tot}}$ returned by Algorithm~\ref{alg:cert_uni} for  $\Theta^i$ is identical to the accumulated complexity measure $\kappa^*_{\text{tot}}$ returned by Algorithm~\ref{alg:BnB_on_uni}, $\forall \theta \in \Theta^i$,  $\forall i$, which completes the proof. 
\end{proof}

\begin{corollary} \label{corr:comp_pwc_uni} 
The complexity measure $\kappa_{\text{tot}}(\theta): \Theta_0 \rightarrow \mathbb{N}_0$ returned by Algorithm~\ref{alg:cert_uni} applied to the problem~\eqref{eq:mpMILP_uni}/\eqref{eq:mpMIQP_uni} is PWC. Furthermore, for the problem~\eqref{eq:mpMILP_uni},  $\kappa_{\text{tot}}(\theta)$ is PPWC.
\end{corollary}

\begin{proof}
This follows directly from Theorem~\ref{thrm:complexity_meas_uni}. Specifically,  $\kappa_{\text{tot}}(\theta)$ remains constant within each region $\Theta^i$ of the partition defined by Algorithm~\ref{alg:cert_uni}, implying that it is PWC.
Furthermore, if all partitioning in Algorithm~\ref{alg:cert_uni} is performed using hyperplanes (i.e., affine functions, as for the problem in~\eqref{eq:mpMILP_uni}), the regions generated at each iteration remain polyhedral. Consequently, $\kappa_{\text{tot}}(\theta)$ becomes PPWC.
\end{proof}


\subsection{Properties of the improvement heuristic certification}
After deriving the properties of Algorithm~\ref{alg:cert_uni},  we can now analyze the properties of the improvement heuristic certification algorithms (where Algorithm~\ref{alg:cert_uni} is used). The proofs of the following lemmas are provided in Appendix.

\begin{lemma}[Equivalence of \textsc{LBCert} and \textsc{LB}]
\label{lem:LB_equiv_uni} 
Assume that Assumptions~\ref{ass:cert_uni}--\ref{ass:order-tree_uni} hold. Then, Algorithm~\ref{alg:heuris_LB_cert_uni} (\textsc{LBCert}) partitions $\Theta$ into subsets $\{\Theta^j\}_j$, such that the potentially found improved solution $\hat{x}^j(\theta)$ and the accumulated complexity measure $\hat{\kappa}^{h_j}(\theta)$ for $\Theta^j$ coincide with $\hat{x}$ and $\hat{\kappa}^{h*}$  returned by Algorithm~\ref{alg:heuris_LB_on_uni} (\textsc{LB}) for any fixed $\theta \in \Theta^j$, $\forall j$.
\end{lemma}

\begin{lemma}[Equivalence of \textsc{RINSCert} and \textsc{RINS}] 
\label{lem:RINS_equiv_uni} 
Assume that Assumptions~\ref{ass:cert_uni}--\ref{ass:order-tree_uni} hold. Then, Algorithm~\ref{alg:heuris_rins_cert_uni} (\textsc{RINSCert}) partitions $\Theta$ into subsets $\{\Theta^i\}_j$, such that the potentially found improved solution $\hat{x}^j(\theta)$ and the accumulated complexity measure $\hat{\kappa}^{h_j}(\theta)$ for $\Theta^j$  coincide with $\hat{x}$ and $\hat{\kappa}^{h*}$ returned by Algorithm~\ref{alg:rins_on_uni} (\textsc{RINS}) for any fixed $\theta \in \Theta^j$, $\forall j$.
\end{lemma}

\section{Conservative complexity certification framework for MIQPs}  \label{sec:cert_cons_uni} 
As discussed before, when certifying the B\&B algorithms for MIQPs, 
comparing quadratic value functions within a parameter set \(\Theta\) (e.g., when evaluating the dominance cut) introduces quadratic inequalities, partitioning \(\Theta\) into non-polyhedral subsets and potentially rendering further computations intractable. To maintain the polyhedral structure of partitions and enhance tractability, spatial partitioning based on quadratic function comparisons can be avoided by conservatively comparing those, thereby accepting some level of conservativeness across the region. 

Methods for conservative quadratic function comparisons are particularly useful when evaluating the dominance cut.  When using the BF node-selection strategy for MIQPs, however, approximating quadratic function comparisons $ \rho_{\text{BF}}(\cdot) $ over $\Theta$  in \textsc{sortCert} can alter the node ordering during insertion into $ \mathcal{T} $, causing deviations from the order determined online by \textsc{sort} for a given $ \theta \in \Theta $. This would lead to different exploration paths and discrepancies between the certification results of Algorithm~\ref{alg:cert_uni} and the behavior of  Algorithm~\ref{alg:BnB_on_uni}.  
To ensure consistency and meaningful certification, we impose the following assumption:
\begin{assumption} \label{ass:BF_MIQP_uni} When Algorithm~\ref{alg:cert_uni}, applied to the problem~\eqref{eq:mpMIQP_uni}, employs an approximated quadratic function comparison (e.g., for the dominance cut evaluation), then either the DF or BrF node-selection strategy is used.
\end{assumption}  

\subsection{Affine approximation of a quadratic function}\label{sec:qp_comp_apprx}
Consider a polyhedral parameter set $\Theta$ and let $\tilde{J}(\theta) = \barbelow{J}(\theta) - \bar{J}(\theta) = \theta^T \tilde{Q} \theta  + \tilde{R} \theta +  \tilde{S}$ 
(see~\eqref{eq:eq_obj_diff_qp_uni}). This section presents three alternative methods for constructing an affine approximation $J'(\theta) = R' \theta + S'$ of $\tilde{J}(\theta)$  within $\Theta$, 
such that $J'(\theta) \geq 0$,  $\forall \theta \in \Theta$,  guarantees that $\tilde{J}(\theta) \geq 0$,  $\forall \theta \in \Theta$. 

\subsubsection{McCormick relaxations}  \label{subsec:qp_comp_mcc_uni}

In this method, the quadratic term in $\tilde{J}(\theta)$ is approximated using McCormick envelopes~\cite{mccormick1976computability}. To this end, new variables $\hat{\theta}_{ij} = \theta_i \theta_j$ are introduced to represent the bilinear terms in $\tilde{J}(\theta)$ for all $i,j$. Since $\theta$ lies within the polyhedral set $\Theta$, each component $\theta_i$ has known bounds: $\barbelow{\theta}_i \leq \theta_i \leq \bar{\theta}_i$. Using these bounds, convex and concave envelopes can be constructed for the bilinear terms, providing linear constraints that relax the original quadratic term. By reformulating $\tilde{J}(\theta)$ into a set of linear inequalities, an affine approximation of the form $J'(\theta, \hat{\theta}) = R' [\theta; \hat{\theta}] + S'$ is obtained, where $R'$ and $S'$ are determined by the coefficients of $\tilde{J}(\theta)$ and the McCormick-envelope bounds.


\subsubsection{Under-approximations of quadratic functions} \label{subsec:qp_comp_aff_uni}
An alternative approach is to construct an affine under-approximation of the quadratic function by minimizing its quadratic term over the region $\Theta$~\cite{arnstrom2021unifying}, given by  
\begin{equation} \label{eq:qp_comp_aff_uni}
J'(\theta) = R' \theta + S' = \min_{\theta \in \Theta}   \left (\theta^T \tilde{Q} \theta \right ) + \tilde{R} \theta +  \tilde{S}.
\end{equation}%
This ensures that $J'(\theta)$ provides a valid approximation, 
while maintaining relatively low computational complexity.


\subsubsection{Regions as atomic units} \label{subsec:qp_comp_single_uni}
A third, relatively simple but more conservative approximation is to approximate the quadratic function $\tilde{J}(\theta)$ using  
\begin{equation} \label{eq:qp_comp_single_uni}
J'(\theta) = S' = \min_{\theta \in \Theta} \left( \theta^T \tilde{Q} \theta + \tilde{R} \theta + \tilde{S} \right),
\end{equation} %
which provides a simple constant approximation. If $S' \geq 0$, then $\tilde{J}(\theta) \geq 0$, $\forall \theta \in \Theta$. This strategy is motivated by treating the region $\Theta$ as an atomic unit when checking whether $\tilde{J}(\theta) \geq 0$, $\forall \theta \in \Theta$, as suggested in previous works~\cite{axehill2010improved,shoja2022overall}. 

\subsection{Conservative cut-condition evaluation} 
A conservative version of Algorithm~\ref{alg:cut_cond_cert_uni} (\textsc{cutCond}) for evaluating the dominance cut while preserving the polyhedral structure is presented in Algorithm~\ref{alg:cut_condcons_cert_uni} (\textsc{cutCondCert}).
In this algorithm, an affine approximation $J'(\theta)$ of $\tilde{J}(\theta) = \barbelow{J}(\theta) - \bar{J}(\theta)$ is first generated using any of the methods in Section~\ref{sec:qp_comp_apprx} (Step~\ref{step:cut_cons_L_uni}). By construction, the condition $J'(\theta) \geq 0$ for all $\theta \in \Theta$ implies that $\tilde{J}(\theta) \geq 0$ for all $\theta \in \Theta$. Consequently, for the subset $\tilde{\Theta} = \{ \theta \in \Theta \mid J'(\theta) \geq 0 \}$ where the dominance cut is conservatively applied (Step~\ref{step:cut_condcons_dom_uni}), the cut remains valid for all $\theta \in \tilde{\Theta}$ under exact evaluation. However, the converse does not necessarily hold, indicating conservativeness. For the remaining region ($\Theta \setminus \tilde{\Theta}$) where $J'(\theta) < 0$, the dominance cut is dismissed, preserving it for further processing.  
To avoid updating the upper bound to a non-improving lower bound at Step~\ref{step:int_cutcons_st_cert_uni} due to conservativeness, it is necessary to maintain a collection of upper bounds, denoted by $\bar{\mathbb{J}}(\theta) = \{\bar{J}^i(\theta)\}_i$ (Step~\ref{step:int_cutcons_cert2_uni}), for each region rather than a single global upper bound $\bar{J}(\theta)$ for the entire region.


\begin{algorithm}[htbp]
\caption{\textsc{cutCertCons}: Conservative cut-condition evaluation and branching in the certification framework} %
\label{alg:cut_condcons_cert_uni}
\begin{algorithmic}[1]
\Require $\left( \Theta, \mathcal{T}, \kappa_{tot}, \bar{\mathbb{J}} \right )$, $\barbelow{J}$, $\barbelow{x}$,  $\mathcal{A}$,  $(\mathcal{B}_0$, $\mathcal{B}_1)$, $\mathcal{S}$ 
\Ensure $\mathcal{S}$ 
\vspace{.06cm}
\If {$ \exists \bar{J} \in \bar{\mathbb{J}}$ and $ \exists \theta \in \Theta$: $ \barbelow{J}(\theta) \geq \bar{J}(\theta)$} \label{step:cut_condcons_cert1_uni} 
\State{Construct an affine approximation $J'(\theta)$ of $\tilde{J}(\theta) = \barbelow{J}(\theta)-\bar{J}(\theta)$, $\forall \theta \in \Theta$}\label{step:cut_cons_L_uni}
\State $\tilde{\Theta} = \{ \theta \in \Theta \;|\; J'(\theta) \geq 0 \}$ \label{step:cut_condcons_dom_uni}
\State {Push $ \left (\tilde{\Theta}, \mathcal{T}, \kappa_{\text{tot}}, \bar{\mathbb{J}}\right )$ to $\mathcal{S}$}
\label{step:dom_cutcons_cert_uni}
\State $\Theta \leftarrow \{\theta \in \Theta \;|\; J'(\theta) < 0 \}$ \label{step:cut_condcons_dom2_uni}
\EndIf
\vspace{-0.1cm}
\If {$\Theta \neq \emptyset$}
\label{step:nodom_cutcons_cert_uni}
\If {all relaxed binary constraints are active} %
\label{step:int_cutcons_cert_uni}
\State $\bar{\mathbb{J}}(\theta) \leftarrow \bar{\mathbb{J}}(\theta) \cup \{J(\theta)\}$ 
\label{step:int_cutcons_cert2_uni}
\State {Push $ \left ( \Theta, \mathcal{T}, \kappa_{\text{tot}}, \bar{\mathbb{J}}\right )$ to $\mathcal{S}$}
\label{step:int_cutcons_st_cert_uni}
\Else \label{step:cut_condcons_branch_uni}
\State{$\{\left(\Theta^k,k \right)\}_{k=1}^{N_b}$ $\leftarrow$ \textsc{branchIndCert}($\barbelow{x}(\theta), \mathcal{B}, \Theta$)} %
\label{step:branch1cons_cert_uni}
\vspace{0.04cm} 
\For {$k \in \mathbb{N}_{1:N_b} $} %
\vspace{0.04cm} 
\State \Longunderstack[l]{$\mathcal{S} \leftarrow$ \textsc{sortCert}$((\mathcal{B}_0 \cup \{k\},\mathcal{B}_1 ), (\mathcal{B}_0,\mathcal{B}_1 \cup \{k\})$, \\$(\Theta^k, \mathcal{T}, \kappa_{\text{tot}}, \bar{\mathbb{J}}), \barbelow{J}(\theta), \mathcal{S} )$} %
\label{step:appendcons_cert_uni}  
\EndFor
\EndIf
\EndIf
\State \textbf{return} $\mathcal{S}$
\end{algorithmic}
\end{algorithm}
Note that in Algorithm~\ref{alg:cut_condcons_cert_uni}, only one of the upper bounds $\bar{J}$ from the collection $\bar{\mathbb{J}}$ is assumed to be selected at Step~\ref{step:cut_cons_L_uni}. However, the algorithm can be modified to loop over all the upper bounds in $\bar{\mathbb{J}}$, potentially partitioning $\Theta$ based on each comparison to eliminate larger parts of $\Theta$ where the dominance cut holds. This would reduce conservativeness at the cost of increased partitioning and potential computational effort.  
In the conservative certification framework in Algorithm~\ref{alg:cert_uni} applied to MIQPs, Step~\ref{step:eval_cut_uni} invokes \textsc{cutCondCons} instead of \textsc{cutCond}, ensuring a polyhedral structure and avoiding nonlinearity within regions.

\subsection{Properties of the conservative certification framework} \label{subsec:prop_cert_cons_uni}
In this section, we analyze the properties of the conservative certification framework to ensure that Algorithm~\ref{alg:cert_uni} provides correct upper bounds on the complexity measures. Specifically, we aim to show that although this conservative strategy may increase the number of explored nodes and potentially introduce additional computations, it still provides a reliable \textit{upper bound} on the worst-case complexity measures. Prior works~\cite{axehill2010improved,shoja2022overall} have shown that an upper bound on the worst-case computational complexity is guaranteed when the dominance cut condition is applied conservatively while treating regions as atomic units.  
We now extend these results to any affine approximation $J'(\theta)$.

From Assumption~\ref{ass:BF_MIQP_uni}, and since the branching indices are determined based on solutions to relaxations—which are also PWA for MIQPs—the properties of \textsc{sortCert} and \textsc{branchIndCert} (established in Lemmas~\ref{lem:append_uni} and~\ref{lem:branch_uni}, respectively) remain unchanged in both the exact and conservative certification frameworks. 
Thus, in this section, we focus on analyzing the properties of  \textsc{cutCertCons} and, subsequently, the conservative framework in Algorithm~\ref{alg:cert_uni}. 

\begin{lemma}[Conservative cut-condition evaluation]
\label{cor:cut_cond_cons_uni}
Assume that Assumptions~\ref{ass:cert_uni}--\ref{ass:BF_MIQP_uni} hold. Then, Algorithm~\ref{alg:cut_condcons_cert_uni} (\textsc{cutCertCons}) partitions the polyhedral set $\Theta$ into polyhedral subsets $\{\Theta^j\}_j$. Moreover, the updated node list $\mathcal{T}^j$ (stored in $\mathcal{S}$) for $\Theta^j$ is a superset of $\mathcal{T}$ returned by Algorithm~\ref{alg:cut_cond_on_uni} (\textsc{cut}) for any fixed $\theta \in \Theta^j$. 
\end{lemma}

\begin{proof}
Partitioning $\Theta$ using the affine approximation $J'(\theta)$ at Steps~\ref{step:cut_condcons_dom_uni} and~\ref{step:cut_condcons_dom2_uni} in \textsc{cutCertCons} divides $\Theta$ into subsets: $\tilde{\Theta} = \{ \theta \in \Theta \mid J'(\theta) \geq 0 \}$ and $ \Theta \setminus \tilde{\Theta} = \{ \theta \in \Theta \mid J'(\theta) < 0 \}$.
Since $J'(\theta)$ is affine and $\Theta$ is polyhedral by assumption, both $\tilde{\Theta}$ and $\Theta \setminus \tilde{\Theta}$ are polyhedral. Furthermore, the partitioning of $\Theta$ under the branching operations at Steps~\ref{step:branch1cons_cert_uni}--\ref{step:appendcons_cert_uni} remains polyhedral by the construction of \textsc{branchIndCert} and Assumption~\ref{ass:BF_MIQP_uni}. Therefore, all partitions generated by \textsc{cutCertCons} are polyhedral.

Next, consider the inputted node $\eta = (\mathcal{B}_0$, $\mathcal{B}_1)$, and assume that the condition at Step~\ref{step:cut_condcons_cert1_uni} is satisfied, i.e., $\exists \bar{J} \in \bar{\mathbb{J}}$ such that $ \barbelow{J}(\theta) \geq \bar{J}(\theta)$. Then, one of the following cases arises for $\eta$ within the inputted $\Theta$: 
\begin{enumerate}[leftmargin=*]
    \item If $J'(\theta) \geq 0$, $\forall \theta \in \Theta$ $\Rightarrow$  
    $ \barbelow{J}(\theta) \geq \bar{J}(\theta)$, $\forall \theta \in \Theta$ (by construction), and  the dominance cut is invoked correctly throughout the entire $\Theta$.
    \item If $\exists \hat{\theta} \in \Theta$:  $J'(\hat{\theta}) < 0$, then:
    \begin{enumerate}
        \item If $\barbelow{J}(\theta) < \bar{J}(\theta)$, $\forall \theta \in \Theta$ $\Rightarrow$ $J'(\theta) < 0$, $\forall \theta \in \Theta$, and the dominance cut is correctly dismissed throughout the entire $\Theta$.
        \item Otherwise, the dominance cut is dismissed incorrectly for all $\theta \in \Theta^{\epsilon}$, where $\Theta^{\epsilon} \triangleq \{\theta \in \Theta \mid J'(\theta) < 0, \; \barbelow{J}(\theta) \geq \bar{J}(\theta)\}$.
    \end{enumerate}
\end{enumerate}

The conservativeness arises in case (2b) for $\Theta^{\epsilon}$, where \textsc{cutCertCons} potentially branches on $\eta$ within $\Theta^{\epsilon}$, whereas $\eta$ would be pruned in \textsc{cut} if evaluated pointwise. The additional descendant nodes generated from branching on $\eta$ are stored in $\mathcal{T}$ for $\Theta^{\epsilon}$ according to Assumption~\ref{ass:order-tree_uni}, whereas the node list would not be updated in \textsc{cut}. Thus, the node list returned by \textsc{cutCertCons} includes additional nodes compared to the list returned by \textsc{cut}, and is, therefore, a superset of this list for any fixed $\theta \in \Theta^{\epsilon}$. 
\end{proof}

\begin{theorem}[Subsequence property of explored node sequences]
\label{thrm:node_seq_cons_uni}
Assume that Assumptions~\ref{ass:cert_uni}--\ref{ass:order-tree_uni} hold. Let $\hat{\mathbb{B}}(\theta)$ denote the sequence of nodes explored  to solve the problem~\eqref{eq:MIQP_uni} for any fixed $\theta \in \Theta_0$ using Algorithm~\ref{alg:BnB_on_uni}  (\textsc{B\&B}). Moreover, let $\mathbb{B}(\theta)$ denote the sequence of nodes explored by Algorithm~\ref{alg:cert_uni} (\textsc{B\&BCert}) while employing \textsc{cutCertCons} applied to the problem~\eqref{eq:mpMIQP_uni} for a terminated region $\Theta^i$ in  $\mathcal{F}$ $= \{(\Theta^i, \kappa_{\text{tot}}^i, \bar{J}^i)\}_i$. Then, $\hat{\mathbb{B}}(\theta) \subseteq \mathbb{B}(\theta)$, $\forall \theta \in \Theta^i$ and all $i$.
\end{theorem}

\begin{proof}
Let $( \Theta, \mathcal{T}, \kappa_{{tot}}, \bar{J})$ be popped from $\mathcal{S}$ at Step~\ref{step:pop_reg_uni} of Algorithm~\ref{alg:cert_uni}, and let a new node $\eta$ be popped from $\mathcal{T}$ at Step~\ref{step:pop_node_cert_uni}. Then, $\eta$ falls into one of two cases: 
\begin{enumerate}[label=(\roman*)]
\item It is a non-redundant node that would be explored online in Algorithm~\ref{alg:BnB_on_uni} for a  fixed $\theta \in \Theta$.
\item It is a redundant node that has been pruned online in Algorithm~\ref{alg:BnB_on_uni} for a fixed $\theta \in \Theta$. In other words, there existed a node $\eta_p$ in a previous iteration where $\eta \in \mathcal{D}(\eta_p)$, such that $\eta_p$ was pruned online in \textsc{cut} due to the dominance cut, whereas it was branched upon in \textsc{cutCertCons} (see case (2b) in the proof of Lemma~\ref{cor:cut_cond_cons_uni}).
\end{enumerate}

For a fixed $\theta \in \Theta$ corresponding to a redundant descendant node in case (ii), there exists a $\bar{J} \in \bar{\mathbb{J}}$ such that $ \barbelow{J}(\theta) \geq \bar{J}(\theta)$.  
From Lemma~\ref{lem:decentent_uni}, the nodes in $\mathcal{D}({\eta_p})$ do not yield a better solution. As a result, exploring these additional descendants does not improve the upper bound for $\theta$. Consequently, the upper bound used in future dominance cuts remains unchanged for $\theta$. Therefore, the exploration of future nodes is not influenced by the exploration of the additional nodes in $\mathcal{D}({\eta_p})$. This implies that Algorithm~\ref{alg:cert_uni} follows a conservative pruning strategy that may explore redundant nodes but does not exclude any relevant nodes explored by Algorithm~\ref{alg:BnB_on_uni}, i.e., $\hat{\mathbb{B}}(\theta) \subseteq \mathbb{B}(\theta)$ for all $\theta$.  
Since $\theta$ was chosen arbitrarily, this holds for all $\theta \in \Theta_0$.
\end{proof}

\begin{theorem}[Upper bound on the complexity measure]
\label{thrm:complexity_meas_cons_uni}
Assume that Assumptions~\ref{ass:cert_uni}--\ref{ass:order-tree_uni} hold, and let $\kappa^*_{\text{tot}}(\theta)$ denote the accumulated complexity measure to solve the problem~\eqref{eq:MIQP_uni} for a fixed $\theta \in \Theta_0$ using Algorithm~\ref{alg:BnB_on_uni}.
Then, the complexity measure $\kappa_{\text{tot}}(\theta): \Theta_0 \rightarrow \mathbb{N}_0$ returned by Algorithm~\ref{alg:cert_uni} while employing \textsc{cutCertCons} applied the problem~\eqref{eq:mpMIQP_uni} for a  region $\Theta^i$ in $\mathcal{F} = \{(\Theta^i, \kappa^i_{\text{tot}}, \bar{J}^i)\}_i$
satisfies $\kappa^i_{\text{tot}} \geq \kappa^*_{\text{tot}}$, $\forall \theta \in \Theta^i$, and all $i$.
\end{theorem}

\begin{proof}
From Assumption~\ref{ass:cert_uni} and Theorem~\ref{thrm:node_seq_cons_uni}, the sequence of nodes, and hence QP relaxations, considered in Algorithm~\ref{alg:cert_uni} for a fixed $\theta$ form a superset of the QP relaxations that would be solved in Algorithm~\ref{alg:BnB_on_uni}. 
As a result, the accumulated complexity measure $\kappa^i_{\text{tot}}(\theta)$ returned by Algorithm~\ref{alg:cert_uni} for $\Theta^i$ upper bounds 
$\kappa^*_{\text{tot}}(\theta)$ returned by Algorithm~\ref{alg:BnB_on_uni}, $\forall \theta \in \Theta^i$, $\forall i$, which completes the proof. 
\end{proof}

\begin{corollary} \label{corr:comp_pwc_cons_uni} 
The complexity measure $\kappa_{\text{tot}}(\theta): \Theta_0 \rightarrow \mathbb{N}_0$ returned by Algorithm~\ref{alg:cert_uni} while employing \textsc{cutCertCons} applied to the problem~\eqref{eq:mpMIQP_uni} is PPWC. 
\end{corollary} 
\begin{proof}
It follows from an analogous reasoning as in the proof of Corollary~\ref{corr:comp_pwc_uni}.
\end{proof}
\section{Extensions} 
\label{sec:extention_uni} 

\subsection{Warm-starting the subproblems} \label{subsec:warmstart_uni}
In the \bnb\ search tree, a sequence of similar relaxations is solved, differing only in the fixed binary variables along a path. To improve runtime, it is standard in \bnb to utilize the solution to the relaxation of a parent node to warm-start the subproblems in its child nodes. In this case, each node retains relevant information from its parent’s solution, particularly the active set $ \mathcal{A} $. 

To enable complexity-certified warm starts for the solver of the relaxations in \bnb-based MILP solvers, a \textit{recovery} procedure is required to extract an initial basis (active set) for the child nodes from $\mathcal{A}$ of the parent node, to drop a constraint from $ \mathcal{A} $. 
A basis recovery procedure for this purpose is provided in~\cite[Algorithm~3]{shoja2022exact} and can be directly incorporated into Algorithm~\ref{alg:cert_uni} to warm-start the solver of the relaxations for MILPs in the certification framework.

For MIQPs, warm-starting is more straightforward, as a recovery approach is not required in this case. The active set  $ \mathcal{A} $ from a parent node can be directly used to initialize and warm-start the solver for the relaxations of the child nodes, enabling complexity-certified warm starts of the solver of the relaxations in \bnb\ for MIQPs using Algorithm~\ref{alg:cert_uni}.

\subsection{Certifying suboptimal \textsc{B\&B} algorithms} %
To reduce the computational complexity of challenging optimization problems, one approach is to relax the requirement for global optimality and accept suboptimal solutions.  
In~\cite{ibaraki1983using}, three suboptimal methods were introduced to reduce the computational burden of B\&B for MILPs by trading optimality for solution time. Building on these methods, the certification framework in~\cite{shoja2023subopt} established worst-case complexity bounds for suboptimal B\&B-based MILP solvers using these suboptimal strategies, summarized below.
\begin{itemize}[leftmargin=*]
\item \textit{The $\epsilon$-optimality method.} This method relaxes the dominance cut $\barbelow{J}(\theta) \geq \bar{J}(\theta)$ by introducing an $\epsilon$-optimality criterion. Specifically, the dominance cut is satisfied if  $\varepsilon(\theta) + (1 + \varepsilon_r)\barbelow{J}(\theta) \geq \bar{J}(\theta)$, 
where \(\varepsilon(\theta) \geq 0\) is generally assumed to belong to the same function class as \(\barbelow{J}(\theta)\) over \(\Theta\), and \(\varepsilon_r \geq 0\) is a constant. 

\item \textit{The $T$-cut method.} This method limits the number of processed nodes to a predefined maximum, terminating the algorithm after $T_0$ node decompositions (branching).

\item \textit{The $M$-cut method.} This method restricts the size of the pending node list \(\mathcal{T}\) to a fixed maximum of $M_0$ nodes. 
\end{itemize}

The $T$-cut and $M$-cut suboptimality methods rely on resource constraints and are independent of whether the objective function is linear or quadratic. Consequently, the results in~\cite{shoja2023subopt} when employing these methods can be directly extended to certify suboptimal \bnb\ algorithms for MIQPs without modification.  
For the $\epsilon$-optimality method, the results in~\cite{shoja2023subopt} can, in principle, be extended to MIQPs if the exact framework is used. In summary, the framework can compute worst-case complexity guarantees for suboptimal B\&B algorithms employing these suboptimality methods for both MILPs and MIQPs. Note however, if the conservative framework from Section~\ref{sec:cert_cons_uni} for MIQPs is combined with the $\epsilon$-optimality method, the subsequence property from Theorem~\ref{thrm:node_seq_cons_uni} is generally no longer preserved. 

\subsection{Certifying the number of floating-point operations}  
\label{subsec:flop_uni}  

The unified certification framework in this work provides detailed insights into the \bnb algorithm by identifying the operations performed, their arguments, and their execution frequency. This allows mapping each operation and its arguments to the corresponding number of floating-point operations (flops).  
Thus, the certification framework can, in principle, be extended to quantify the number of flops required by the online \bnb algorithm to compute optimal or suboptimal solutions for all parameters of interest. To this end, it is necessary to track \emph{all} operations executed during the certification process of the relaxations in \textsc{solveCert}, as well as those performed throughout the \bnb algorithm.

\subsection{Extension to general mixed-integer optimization problems} 
\label{subsec:general_prob_uni}  
The certification framework in Algorithm~\ref{alg:cert_uni} has the potential to be extended to more general (multi-parametric) mixed-integer optimization problems, provided that (1) the \textsc{solveCert} function (at Step~\ref{step:cert_node_uni}) is available for other types of relaxations, and (2) value function comparisons (e.g., at Steps~\ref{step:cut_cond_cert1_uni} and~\ref{step:cut_cond_cert2_uni} of Algorithm~\ref{alg:cut_cond_cert_uni}) can be performed for objective functions beyond linear and quadratic functions.
\section{Numerical experiments}
In this section, the proposed complexity certification framework is applied to random MILPs, MIQPs, and optimization problems originating from an MPC application. The complexity measures considered include the number of iterations ($\kappa^I$) and the number of \bnb nodes ($\kappa^N$). For the \textsc{solveCert} procedure, the certification algorithm for LPs/QPs described in~\cite{arnstrom2021unifying} was utilized. Furthermore, to ensure tractable computations when certifying \bnb-based MIQP solvers, the conservative framework with the affine approximations in~\eqref{eq:qp_comp_aff_uni} and~\eqref{eq:qp_comp_single_uni} was applied. 
The numerical experiments were implemented in Julia (version 1.10.5) and executed on a laptop with an Intel$^{\circledR}$ Core i7-8565U CPU at 1.8 GHz. 

\subsection{Random examples}
First, we consider randomized mp-MILPs and mp-MIQPs with $n_b = n_c =8$ binary and continuous decision variables, 
$m = 16$ constraints, and $n_{\theta} = n_b/2 $ parameters. 
The coefficients of~\eqref{eq:mpMILP_uni} and~\eqref{eq:mpMIQP_uni} were generated as: $\bar{H}~\sim~\mathcal{N}(0,1)$, $H~=~\bar{H} \bar{H}^T$, $f~\sim~\mathcal{N}(0,1)$, $c~\sim~\mathcal{N}(0,1)$, $f_{\theta}~\sim~\mathcal{N}(0,1)$, $A~\sim~\mathcal{N}(0,1)$, $b~\sim~\mathcal{U}([0,2])$, and  $W~\sim~\mathcal{N}(0,1)$, 
where \(\mathcal{N}(\mu, \sigma)\) represents a normal distribution with mean \(\mu\) and standard deviation \(\sigma\), and \(\mathcal{U}([l, u])\) denotes a uniform distribution over the interval \([l, u]\). The parameter set was defined as $\Theta_0 = \{\theta \in \mathbb{R}^{n_{\theta}} \mid -0.5 \leq \theta_i \leq 0.5, \forall i\}$.

The certification framework in Algorithm~\ref{alg:cert_uni} was validated 
on 10,000 samples drawn from \(\Theta_0\) on a deterministic grid, along with the Chebyshev centers of the resulting terminated regions. The online \textsc{B\&B} Algorithm~\ref{alg:BnB_on_uni} was applied to the corresponding MILP and MIQP instances. For MILPs, 
the complexity measures \(\kappa^I\) and \(\kappa^N\) were identical across all sampled points and problems, validating the correctness of the certification framework. For MIQPs, the experiments confirmed that the conservative certification framework provided valid upper bounds on the worst-case complexity measures without underestimation. The worst-case iteration number $\kappa^I$ was overestimated by approximately $8\%$ when using the affine approximation in~\eqref{eq:qp_comp_single_uni}, but this conservativeness decreased to $4\%$ when using the affine approximation in~\eqref{eq:qp_comp_aff_uni}.

To demonstrate the potential use of the results presented in this work, Tab.~\ref{tab:cert_res_uni} summarizes the results of applying the certification framework to 50 random mp-MILPs, averaged over all experiments with different warm-started B\&B configurations. 
The table presents complexity metrics, including the worst-case (wc) and average (avg) number of iterations and nodes, as well as the number of final regions (\#reg) and the CPU time required for certification (\(t_{\text{cert}}\)). The results highlight that the BF and DF strategies generally require fewer computations compared to BrF.  Additionally, the choice of branching strategy (first binary (FB) variable vs. MIB) significantly impacts the results, with MIB typically leading to fewer computations. For random mp-MIQPs, the same trend was observed, with BrF and FB settings demanding considerably higher computational effort. While the purpose of this work is not to determine the superiority of a particular B\&B setting, it provides tools that enable answering that question for the specific application.

\begin{table}[htbp]
\centering
\caption{Certification results for random mp-MILPs.}
\scalebox{0.94}{ 
\begin{tabular}{l|cccccc}
B\&B setting & 
$\kappa^I_{\text{wc}}$ & $\kappa^N_{\text{wc}}$ & $\kappa^I_{\text{avg}}$ & $\kappa^N_{\text{avg}}$ & $\#\text{reg}$ & $t_{\text{cert}}$[sec]  \\ \midrule 
DF, FB             & 75.78  & 45.17  & 64.93   & 32.9  & 775  & 12.7 \\ 
DF, MIB            & 45.28  & 8.78   & 42.19   & 6.78  & 388    & 8.3  \\ 
BrF, FB            & 106.94 & 76.61  & 91.41   & 56.08 & 2918 & 20.7 \\ 
BrF, MIB           & 46.83  & 9.67   & 43.44   & 7.46  & 482   & 9 \\ 
BF, FB             & 59.56  & 25.89  & 52.49   & 18.46 & 706  & 11.9  \\ 
BF, MIB            & 45.61  & 8.83   & 42.09   & 6.44  & 416   & 8.1 
\end{tabular}
}
\label{tab:cert_res_uni}
\end{table}

To give some visual insights into the results of Algorithm~\ref{alg:cert_uni}, a random MILP in a two-dimensional parameter space was certified using the BF, MIB, and warm-started settings in B\&B. Fig.~\ref{fig:example_rand_uni}~(a) illustrates the resulting accumulated number of iterations, where regions with the same number of iteration count share the same color. 
To demonstrate the online performance, the simulation results were obtained by sampling $\Theta_0$ on a deterministic grid and applying Algorithm~\ref{alg:BnB_on_uni} to the corresponding MILP for each sample. The accumulated iteration numbers for the sample points are depicted by $*$ in Fig.~\ref{fig:example_rand_uni}~(b). As the figure indicates, the accumulated iteration numbers from the certification and online results coincide exactly for all parameters.

\begin{figure}[htbp]  
\begin{subfigure}{0.47\columnwidth}		
\centerline{
\includegraphics[scale=0.58]{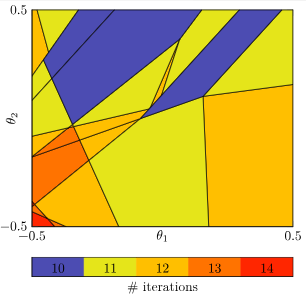}}	\caption{certification} 
\end{subfigure}
\hfill
\begin{subfigure}{0.47\columnwidth}	
\centerline{
\includegraphics[scale=0.58]{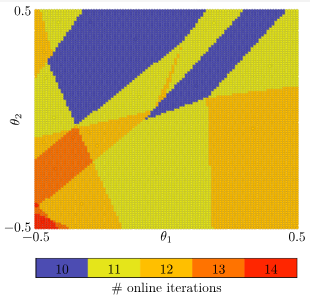}}	
\caption{online  samples} 
\end{subfigure} 
\caption{Resulting parameter space for a random example determined by (a) applying  Algorithm~\ref{alg:cert_uni}; (b) executing Algorithm \ref{alg:BnB_on_uni} over a deterministic grid in the parameter space.}
\label{fig:example_rand_uni}
\end{figure}

\subsection{MPC application}
Next, we apply the proposed method to a hybrid MPC example to demonstrate its applicability. The application involves a linearized inverted pendulum on a cart, constrained by a wall, which introduces contact forces into the system (see Fig.~\ref{fig:pend_uni}). This scenario is used in works such as~\cite{marcucci2020warm} and~\cite{arnstrom2023bnb}. The objective is to stabilize the pendulum at the origin ($z_1 = 0$) while maintaining its upright position ($z_2 = 0$). The control input is a force $u_1$ exerted directly on the cart. When the tip of the pendulum contacts the wall, a contact force $u_2$ is generated, necessitating the inclusion of binary variables in the model. To account for this, two binary variables are introduced into the model. The control inputs were constrained by $|u_1| \leq 1$, $u_2 \geq 0$, and $u_3, u_4 \in \{0, 1\}$, while the state constraints were $|z_1| \leq d$, $|z_2| \leq \frac{\pi}{10}$, $|z_3| \leq 0.5$, and $|z_4| \leq 0.5$, with $z_3 = \dot{z_1}$, $z_4 = \dot{z_2}$, and $d = 0.5$. An MPC with a quadratic performance measure, a prediction/control horizon of $N$, and a time step of $0.1$ seconds was employed. The (initial) state vector $z = [z_1, z_2, z_3, z_4]^T$ formed the parameter vector $\theta$ within the multi-parametric framework, and the parameter set $\Theta_0$ was determined by the state constraints.
The resulting mp-MIQP problem (using the formulation in, e.g.,~\cite{borrelli2017predictive})  contains $n = 4N$ decision variables, of which $2N$ are binary. The weight matrices and dynamics were consistent with those used in~\cite{marcucci2020warm}. For further details, see~\cite{marcucci2020warm}. 

\begin{figure}[htbp]  	
\centerline{	\includegraphics[scale=0.23]{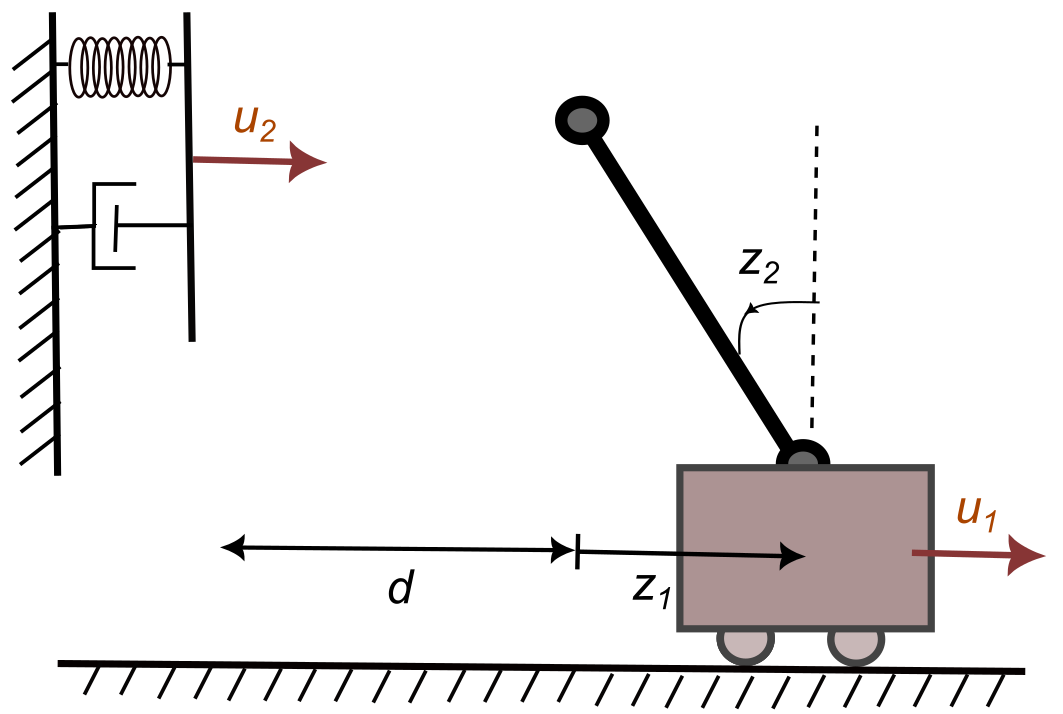}}  
\caption{Regulating the inverted pendulum on a cart with contact forces.} 
\label{fig:pend_uni}
\end{figure} 

Tab.~\ref{tab:pend_MIQP_uni} summarizes the complexity certification results obtained by applying Algorithm~\ref{alg:cert_uni} to the resulting mp-MIQP from MPC for the inverted pendulum with varying horizons $N$, while employing the affine approximations~\eqref{eq:qp_comp_single_uni} for quadratic function comparisons. 
The table reports the number of decision variables $n$, constraints $m$, the resulting complexity measures, and certification time. 
As $N$ increases, the problem size increases, and consequently, the computations required grow significantly. Larger horizons result in a noticeable rise in the number of regions and certification time, illustrating the combinatorial complexity of the problem. These results highlight the framework's ability to handle real-world applications, even when executed on a standard laptop with limited computational resources.

\begin{table}[htbp]
\centering
\caption{
Dimensions of the resulting mp-MIQPs for the inverted pendulum example, along with the worst-case and average number of iterations  ($\kappa^I_{\text{wc}}$ and $\kappa^I_{\text{avg}}$) and B\&B nodes ($\kappa^N_{\text{wc}}$ and $\kappa^N_{\text{avg}}$), the number of final regions ($\#\text{reg}$), and the  time ($t_{\text{cert}}$) taken by the Julia implementation of Algorithm~\ref{alg:cert_uni}.} 
\label{tab:pend_MIQP_uni}
\scalebox{0.96}{
\begin{tabular}{@{}ccccccccc@{}}
\toprule
$N$    &  $n$    &  $m$    & 
$\kappa^I_{\text{wc}}$ & $\kappa^N_{\text{wc}}$ & $\kappa^I_{\text{avg}}$ & $\kappa^N_{\text{avg}}$ & $\#\text{reg}$ & $t_{\text{cert}}$[sec] \\ \midrule
1   & 4  &  19  
& 9               & 4                & 5.7             & 2.9          & 14      & 1.4 \\
2    & 8  &  38      & 22               & 8               & 17.5           & 7.9          & 56      & 6.2 \\
3    & 12  &  57      & 76               & 47               & 44.2           & 26         & 101     & 44.6 \\
4    & 16  &  76      & 134              & 68               & 66.1           & 21.7         & 493     & 122 \\ 
5    & 20  &  95      & 198              & 76               & 76.6           & 21.9        & 1445
& 635 \\ 
\bottomrule
\end{tabular}
}
\end{table}

\section{Conclusion}
This paper presents a framework that extends and unifies complexity certification results for standard \bnb algorithms applied to MILPs and MIQPs. The proposed method allows for an exact analysis of computational complexity, in terms of e.g., the number of iterations and/or B\&B nodes as functions of the parameters in an mp-MILP or mp-MIQP problem, providing precise insights into the computations required when the problem is solved online.  
The framework partitions the parameter space into regions, each corresponding to parameter sets that generate identical solver state sequences for reaching a solution. 
The framework was extended to include different branching strategies, node-selection strategies, and heuristics. 
For MIQPs, quadratic function approximations were developed to handle non-polyhedral parameter-space partitions. These approximations result in good upper bounds on the worst-case computational complexity for MIQPs. 
By providing theoretical guarantees alongside practical insights, the framework enhances the reliability and performance of \bnb-based solvers, which are essential for real-time applications such as hybrid MPC. 
Future work will focus on exploring parallel implementations of the unified certification framework to be able to certify even more challenging problem instances. 

\section*{Appendix: Proof of lemmas from Section~\ref{subsec:prop_cert_uni}} \label{append:proof_uni} 
\subsubsection{Proof of Lemma~\ref{lem:append_uni} \textit{[Equivalence of \textsc{sortCert} and \textsc{sort}]}}
Given Assumptions~\ref{ass:cert_uni}--\ref{ass:order-tree_uni}, the sorting criterion $ \rho(\cdot) $ is identical in both procedures, pointwise. It is also assumed that the inputted $\mathcal{T}$ is sorted according to $ \rho(\cdot) $.
At each iteration of the for-loop in \textsc{sortCert}, $ \Theta $ is potentially partitioned into two subsets for each node $ \eta_i $:  
\begin{equation} \label{eq:part_sort_uni1}
\Theta^i = \{\theta \in \Theta \mid \rho(\tilde{\eta}_0(\theta)) \leq \rho(\eta_i(\theta))\},
\end{equation}
and the remaining region,  
\begin{equation} \label{eq:part_sort_uni2}
\Theta = \{\theta \in \Theta \mid \rho(\tilde{\eta}_0(\theta)) > \rho(\eta_i(\theta))\}.
\end{equation}  
New nodes $ \tilde{\eta}_0 $ and $ \tilde{\eta}_1 $ are then inserted into $ \mathcal{T}^i $ at position $ i $, ensuring $ \rho(\tilde{\eta}_0) \leq \rho(\eta_i) $ for all $ \theta \in \Theta^i $ (Step~\ref{step:append_off_updateT_uni}), thereby preserving the sorting order. The subset $ \Theta^i $ and the list $ \mathcal{T}^i $ are stored in $ \mathcal{S} $ at Step~\ref{step:append_off_push_uni}, while the process continues for the remaining region $ \Theta $ until it is fully partitioned.  

Now, consider \textsc{sort} for any fixed $ \theta \in \Theta^i $. Since $ \rho(\cdot) $ is evaluated pointwise, it inserts $ \tilde{\eta}_0 $ and $ \tilde{\eta}_1 $ into $ \mathcal{T} $ at the first position $ i $ satisfying $ \rho(\tilde{\eta}_0) \leq \rho(\eta_i) $ (Step~\ref{step:append_on_pushT_uni}), identical to Step~\ref{step:append_off_updateT_uni} in \textsc{sortCert}. Since the additional Step~\ref{step:append_off_push_uni} in \textsc{sortCert} only records $ \mathcal{T}^i $ in $ \mathcal{S} $ without modifying it, the final node lists remain unchanged. 
Thus, \textsc{sortCert} partitions $ \Theta $ while preserving the same node ordering as \textsc{sort}. As a result, the resulting node lists $ \mathcal{T}^i $ and $ \mathcal{T} $ are identical for  $\theta$.  Since $\theta$ and hence $i$ are arbitrary, it holds for all $ \theta \in \Theta^i$, and all $ i$, completing the proof. \hfill \textbf{\(\blacksquare\)}

\subsubsection{Proof of Lemma~\ref{lem:most_inf_score_uni} \textit{[Equivalence of \textsc{mostInfScoreCert} and~\eqref{eq:branch_minf_on_uni}]}}
Consider an index $i \in \bar{\mathcal{B}}$ selected at Step~\ref{step:minf_cert_j_uni} in \textsc{mostInfScoreCert}. At Step~\ref{step:minfb_partitioned_cert}, the current set $\Theta$ is partitioned along $\barbelow{x}_i(\theta) = \barbelow{F}_i \theta + \barbelow{g}_i = 0.5$ into subsets $\Theta^{i_1}$ and $\Theta^{i_2}$. Because $\barbelow{x}_i(\theta)$ is affine, an affine score function within each subset is directly obtained, with:
\begin{itemize}[leftmargin=*]
    \item In $\Theta^{i_1}$, $s_i(\theta) = \barbelow{x}_i(\theta)$, with $(C_i, d_i) = (\barbelow{F}_i, \barbelow{g}_i)$ (Step~\ref{step:minfb_Theta1_cert}).
    \item In $\Theta^{i_2}$, $s_i(\theta) = -\barbelow{x}_i(\theta) + 1$, with $(C_i, d_i) = (-\barbelow{F}_i, -\barbelow{g}_i + 1)$ (Step~\ref{step:minfb_Theta2_cert}).
\end{itemize}
Thus, the score function satisfies:
\[
s_i(\theta) = 0.5 - |\barbelow{x}_i(\theta) - 0.5|, \quad \forall \theta \in \Theta^{i_1} \cup \Theta^{i_2} \subseteq \Theta.
\]
confirming that $s_i(\theta)$ satisfies~\eqref{eq:branch_minf_on_uni} for any fixed $\theta \in \Theta$.
Since $i$ was chosen arbitrarily, this holds for all $i \in \bar{\mathcal{B}}$.
By iterating through all indices, \textsc{mostInfScoreCert} partitions $\Theta$ into subsets $\{\Theta^j\}_j$, each associated with an affine score function $(C^j, d^j)$ satisfying~\eqref{eq:branch_minf_on_uni}. Therefore, Algorithm~\ref{alg:mostinfeas_cert_uni} correctly computes the MIB score function in~\eqref{eq:branch_minf_on_uni} pointwise.
 \hfill \textbf{\(\blacksquare\)}

\subsubsection{Proof of Lemma~\ref{lem:branch_uni} \textit{[Equivalence of \textsc{branchIndCert} and \textsc{branchInd}]}}
By Assumption~\ref{ass:cert_uni}, the correctly computed solution to the relaxation \( \barbelow{x}(\theta) \) results in an identical candidate set \( \bar{\mathcal{B}} \) in both procedures.  At each iteration of the for-loop in \textsc{branchIndCert}, the parameter set \( \Theta \) is partitioned into subsets based on the highest score condition:  
\begin{equation} \label{eq:part_branch_uni}
\Theta^k = \{\theta \in \Theta \mid s_k(\theta) \geq s_i(\theta), \;\forall i \in \bar{\mathcal{B}} \setminus \{k\} \},
\end{equation}
ensuring that for any \( \theta \in \Theta^k \), $k \in \bar{\mathcal{B}}$ gives the highest score. The subset \( \Theta^k \) and $k$ are then stored in \( \mathcal{F}_b \) at Step~\ref{step:branch_pushF_cert_uni}, while the process continues for the remaining region of $\Theta$ until it is fully partitioned. 

Now, consider \textsc{branchInd} for a fixed \( \theta \in \Theta^k \). 
Since \( s_i(\theta) \) is evaluated pointwise, \textsc{branchInd} selects \( k \) as the index that maximizes \( s_i(\theta) \) among all \( i \in \bar{\mathcal{B}} \) at Step~\ref{step:select_branch_on_uni}, identically to the selection at Step~\ref{step:branch_partition_cert_uni} in \textsc{branchIndCert}.
 Thus, the chosen branching index \( k \) is the same in both procedures for all \( \theta \in \Theta^k \) 
and all $k$, completing the proof. \hfill \textbf{\(\blacksquare\)}

\subsubsection{Proof of Lemma~\ref{lem:cut_cond_uni} \textit{[Equivalence of \textsc{cutCert} and \textsc{cut}]}}
Given Assumption~\ref{ass:cert_uni}, the solution to the relaxation and the active set are correctly computed at Step~\ref{step:cert_node_uni} in Algorithm~\ref{alg:cert_uni}. Both procedures begin by evaluating the dominance cut at Step~1.  
In \textsc{cutCert}, the dominance cut is evaluated over $ \Theta $, leading to the partitioning of $ \Theta $ into two subsets:  
\begin{equation} \label{eq:part_cut_uni1}
\tilde{\Theta} = \{\theta \in \Theta \mid \barbelow{J}(\theta) \geq \bar{J}(\theta)\},
\end{equation}
and the remaining region,  
\begin{equation} \label{eq:part_cut_uni2}
\Theta = \{\theta \in \Theta \mid \barbelow{J}(\theta) < \bar{J}(\theta) \}.
\end{equation}

For $ \tilde{\Theta} $, the dominance cut is applied at Step~\ref{step:dom_cut_cert_uni} in \textsc{cutCert}, corresponding to Step~\ref{step:dominc_on_uni} in \textsc{cut} for any fixed $ \theta \in \tilde{\Theta} $. In \textsc{cutCert}, the B\&B tree is cut at this node for all $ \theta \in \tilde{\Theta} $, and this subset is excluded from further branching and is pushed to $\mathcal{S}$ without updating $\mathcal{T}$ or $\bar{J}(\theta)$ (Step~\ref{step:dom_cut_cert_uni}). Analogously, in \textsc{cut}, further branching is excluded at Step~\ref{step:dominc_on2_uni}, returning $ \mathcal{T} $ and $ \bar{J} $ without updates. This ensures identical results  $ \forall \theta \in \tilde{\Theta} $.  

For the remaining region $ \Theta $, the integer-feasibility cut is checked at Step~\ref{step:int_cut_cert_uni} in \textsc{cutCert}, corresponding to Step~\ref{step:int_cut_on_uni} in \textsc{cut} for any fixed $ \theta \in \Theta $. Given Assumption~\ref{ass:cert_uni}, the correctly computed active set (which is fixed on $\Theta$) leads to the identical invocation of the integer-feasibility cut by \textsc{cutCert} and \textsc{cut} for $\theta$. If the cut holds, then $ \bar{J}(\theta) $ is updated to $ \barbelow{J}(\theta) $ at Step~\ref{step:int_cut_cert2_uni}, identically to Step~\ref{step:intfeas_on_uni} in \textsc{cut} where $ \bar{J} $ is updated to $ \barbelow{J} $. The new tuple is pushed to $ \mathcal{S} $ without modifying $ \mathcal{T} $, and both B\&B trees are cut at this node in \textsc{cutCert} and \textsc{cut}, ensuring identical results for all $ \theta \in \Theta $.  

If no cut condition holds for $ \Theta $, branching is initiated by selecting the branching index using \textsc{branchIndCert} at Step~\ref{step:branch1_cert_uni} for $ \Theta $ in \textsc{cutCert} and using \textsc{branchInd} at Step~\ref{step:branch_on_uni} for any fixed $ \theta \in \Theta $ in \textsc{cut}. By Lemma~\ref{lem:branch_uni}, \textsc{branchIndCert} partitions $ \Theta $ into $\{\Theta^k\}_k$ such that the selected index $ k $ for $ \Theta^k \ni \theta $ coincides with that chosen by \textsc{branchInd} for $ \theta $.  
We now focus on $ \Theta^k \ni \theta $. New nodes created by fixing binary variable $ k $ are stored in the node list using \textsc{sortCert} at Step~\ref{step:append_cert_uni} in \textsc{cutCert} within $ \Theta^k$, and using \textsc{sort} at Step~\ref{step:append_on_uni} in \textsc{cut} for $\theta$. By Lemma~\ref{lem:append_uni}, \textsc{sortCert} further partitions $ \Theta^k $ into $\{\Theta^j\}_j$, ensuring that the node ordering in $ \mathcal{T}^j $ for $ \Theta^j \ni \theta $ remains identical to that by \textsc{sort} for $\theta$.  That is, the node lists returned by both procedures coincide for $\theta$. 
Since $ \theta $ is arbitrary, this holds for all $\theta \in \Theta$, completing the proof.  \hfill \textbf{\(\blacksquare\)} 

\subsubsection{Proof of Lemma~\ref{lem:LB_equiv_uni}} \textit{[Equivalence of \textsc{LBCert} and \textsc{LB}]}
Using the same reasoning as in the proof of Theorem~\ref{thrm:node_seq_uni}, we analyze an iteration of \textsc{LBCert} for an arbitrary \( \theta \in \Theta   \), starting at Step~\ref{step:LB_pop_cert}, and the corresponding iteration in \textsc{LB} for \( \theta \), starting at Step~\ref{step:LB_Nmax_on}. 
At Step~\ref{step:LB_cert_uni}, the sub-MILP is certified for all \( \theta \in \Theta \) in \textsc{LBCert}, whereas in \textsc{LB}, it is solved online for  \( \theta \) at Step~\ref{step:LB_on_solve}. Since Algorithm~\ref{alg:cert_uni} (\textsc{B\&BCert}) is used for certification and Algorithm~\ref{alg:BnB_on_uni} (\textsc{B\&B}) serves as the online solver, the correctness of their pointwise equivalence follows directly from Theorems~\ref{thrm:node_seq_uni} and~\ref{thrm:complexity_meas_uni}. Consequently, the accumulated complexity measure \( \hat{\kappa}^{h_j}(\theta) \) in \textsc{LBCert} for \( \Theta^j \ni \theta \) (updated at Steps~\ref{step:LB_cert_F} and~\ref{step:LB_cert_S}) is identical to the accumulated complexity measure \( \hat{\kappa}^{h*} \) in \textsc{LB} (updated at Step~\ref{step:LB_on_kappa}) for \( \theta \). 

\textsc{LBCert} then proceeds to one of the two following updates for \( \Theta^j \ni \theta \):
(1) If a new improved solution \( \hat{x}^j(\theta) \) is found, it is stored in the output list \( \mathcal{F}_h \) (Step~\ref{step:LB_cert_F}).
(2) If no better solution is found, $r_n$ is updated using the \textsc{neighborSize} procedure (Step~\ref{step:LB_kneigh_cert}), and the results are pushed to \( \mathcal{S}_h \) (Step~\ref{step:LB_cert_S}). 
These steps are analogous to those in \textsc{LB}. Specifically, for \( \theta \), either:
(1) a new solution is found, and the algorithm terminates (Step~\ref{step:LB_on_sol}), or
(2) $r_n$ is updated using the same \textsc{neighborSize} procedure (Step~\ref{step:LB_kneigh_on}). 
Thus, the results of both algorithms coincide at the end of the iteration for \( \theta \). The same reasoning extends to any \( \theta \in \Theta \), completing the proof. \hfill \textbf{\(\blacksquare\)}

\subsubsection{Proof of Lemma~\ref{lem:RINS_equiv_uni} \textit{[Equivalence of \textsc{RINSCert} and \textsc{RINS}]}}
    The one-to-one correspondence between Steps~\ref{step:rins_on_B_heu}--\ref{step:rins_cert_uni} in \textsc{RINSCert} and \textsc{RINS} 
    follows directly from their structural similarity. At Step~\ref{step:rins_cert_uni}, the sub-MILP is certified for all \( \theta \in \Theta \) in \textsc{RINSCert}, 
    whereas in \textsc{RINS}, 
    it is solved online for the fixed \( \theta \) at Step~\ref{step:rins_on_solv_heu}. Since Algorithms~\ref{alg:cert_uni} and~\ref{alg:BnB_on_uni} are used, respectively,  at these steps,
the desired results follow directly from Theorems~\ref{thrm:node_seq_uni} and~\ref{thrm:complexity_meas_uni}. \hfill \textbf{\(\blacksquare\)}

		\bibliography{IEEEabrv,references} 
			
\end{document}